%% file: oopsla2027.tex
\pdfoutput=1
\documentclass[acmsmall,screen,nonacm]{acmart}
\setcopyright{none}
\acmDOI{}

\input{preamble}
\title{\papertitle}

\begin{abstract}
In wireless sensor and mesh networks, modular robots, and shape-adaptive
computers, computing modules attach, detach, or are replaced while the others
keep running with their local state. A value computed under
the old arrangement, such as a cached next hop, stays in memory but may no
longer describe the current one, so a program can send to a neighbor
that has left without any runtime errors.

We present \emph{morph-aware computing}, a programming model in which the
application specifies a corrective update of such state, called a
\emph{repair}, and the compiler decides when it runs. In \emph{MorphLang}, the
repair is a \texttt{Reconfigure} branch next to the ordinary branches that
handle sends, receives, and sensor samples. \emph{MorphC} compiles MorphLang to
bare-metal RISC-V and inserts a guard before every ordinary branch, so that each
module repairs its state before any other branch sees a new topology. Relying on
this order, MorphC also rejects sends whose target may be a literal or a value
the repair does not always refresh.

For a model of the code MorphC generates, we prove in Rocq that the repair
precedes every later ordinary branch for any number of roles and topology
changes. This shows that send-target checking is sound
for a core language. On 19 MorphLang programs derived from seven network
stacks, removing the repair worsens destination validity in ten of the 13
programs that declare one and never improves it. In a reduced reproduction of
a released Contiki-NG defect, a callback-based version sends to a departed
parent after 60 of 384 changes and the MorphLang version after none. The guards
cost 12 cycles per loop in a microbenchmark, and on an FPGA the delay from
each commit to the first send to the new neighbor matches RTL simulation.
\end{abstract}


\begin{document}
\maketitle

\section{Introduction}
\label{sec:introduction}

\input{figures/fig-teaser-float}

Many systems change their physical configuration while running. They consist
of \emph{computing modules}, physical objects that compute and communicate, and
these modules can move, connect, disconnect, or be replaced. Examples are wireless sensor and mesh
networks~\cite{levis2005tinyos,oikonomou2022contiking,openthread}, modular
robots~\cite{goldstein2005matter,derosa2008ldp}, and shape-adaptive computers
such as millimeter-scale wireless chiplets~\cite{kadomoto2019wixi,kadomoto2020riscv},
energy-harvesting CMOS systems~\cite{xu2022microrobot}, and electrostatically
moving chips~\cite{peng2023electrostatic}.
Such a change in the physical arrangement can alter which modules can
communicate directly, that is, the communication topology. We call a change in
this topology, the only part a program observes, a topology change.
Modules that remain after such a change are not restarted. They keep
executing and retain their local memory.


A topology change does not affect the state a program computed before it, so
that state can outlive the topology it described.
Figure~\ref{fig:teaser} shows an example. A sensor caches a
\texttt{parent} that identifies an adjacent relay. After that relay is replaced,
the cached integer remains in memory even though it no longer points to a neighbor, so the
next send is dropped. Polling the neighbor count does not help because
it stays at one across the replacement. We call such state \emph{stale} when it
no longer matches the current topology. In the issue trackers of
released network stacks, we found 13 reports from four
stacks in which a stale value survived a topology change and caused a failure
(Section~\ref{sec:field-evidence}). In ten of them, the stack did report the
change, and the failure lay in what the application updated and when.

Who corrects such state, and when? A runtime can report a topology change,
but it cannot decide which values to retain, reset, or recompute, so the
decision belongs to the application. We call this application-defined
corrective update a \emph{repair}. In the example, the repair sets \texttt{parent}
to a current neighbor. A repair must
run after a topology change and before the state is next used. Enforcing this
order by hand means adding a check to every event handler, and omitting the
check in any one handler reintroduces the failure.

We propose \emph{morph-aware computing}, a programming model in which
the application writes the repair, and the compiler arranges that it runs
before the program's next ordinary branch. We call this guarantee the
\emph{ordered-repair contract}.

We implement the model with \emph{MorphLang} and \emph{MorphC}. A MorphLang
program defines one role per module. A role consists of \emph{ordinary branches}, which handle sends, receives, and
sensor samples, and a \texttt{Reconfigure} branch that holds the repair. MorphC
compiles the program to bare-metal RISC-V RV32I and inserts a \emph{revision guard}
before every ordinary branch. The guard compares a counter that advances at
each topology change, the \emph{revision}, with the last revision the program
has repaired for, and it runs the repair first if they differ. Because the
branches are part of the language rather than a runtime library, MorphC sees
every ordinary branch and can guard all of them.
MorphC can also rely on the repair running first. Its
\emph{send-target checking} uses this to reject numeric destinations that may
be a literal core identifier or persistent state that no repair path refreshes.
Without the ordering, this check would be unsound because an ordinary branch
could send before the refresh runs (Section~\ref{sec:send-target-validation}).

Related mechanisms fall short in different ways (Section~\ref{sec:related-work}).
nesC and P~\cite{gay2003nesc,desai2013p} can express the update as a handler but
leave its order relative to later handlers to the program. Choral~\cite{giallorenzo2024choral}
assumes reliable communication, which a send to a departed neighbor violates.
Virtual synchrony and ULFM~\cite{birman1987virtual,bland2013ulfm} order or agree
on membership changes for a whole group. Aggregate languages and dynamic software
updating~\cite{viroli2019aggregate,hicks2005dsu} absorb a change by re-evaluating
every round or by applying a state transformer once per patch. None of these mechanisms orders an
application update against a module's own later branches for each topology change.

Earlier work presented a functional reactive prototype of MorphLang with C++
translation~\cite{izawa2023functional}, the reactive language design from which the
ordinary branches come~\cite{izawa2024reactive}, and the VisMorph editor
and inspector~\cite{izawa2025vismorph}. It did not define when application
state must be repaired relative to later computation, nor did it provide a
compiler that enforces this order on the target architecture. This paper closes
both gaps with a programming model, a compiler, formal validation, and an
experimental evaluation.

The contributions of this paper are as follows.

\textbf{A programming model, language, and compiler}
(Sections~\ref{sec:morphlang} and~\ref{sec:semantics-implementation}).
MorphLang repairs an explicit branch.
MorphC generates revision guards, rejects the numeric send targets described
above, and emits bare-metal RV32I.

\textbf{A formal contract with mechanized guarantees}
(Sections~\ref{sec:semantics-implementation} and~\ref{sec:validation}).
We define the ordered-repair contract operationally. For a model of the generated code,
Rocq (formerly Coq)~\cite{coq} proves that
every ordinary branch runs after the repair for its revision, for any number of
roles, and that each delivered revision is repaired at most once. It also
proves send-target checking sound for a core language. TLA+~\cite{lamport2002tla}
and a binary checker of the generated binaries complement the proofs.

\textbf{An evaluation of repair, compiler assistance, and cost}
(Section~\ref{sec:evaluation}).
On 19 hand-written MorphLang programs derived from 13 routines of seven released
network stacks, run
in the cycle-level simulator \emph{MorphSim}, removing repair lowers
destination validity in ten of the 13 programs that declare a repair and never
raise it. Omitting the first generated guard of each program increases sends
to departed relays, and the guards cost
12 cycles per two-branch loop. A reduced reproduction of a Contiki-NG ordering
bug sends to a departed parent after 60 of 384 changes with a callback and
after none with MorphC. On an FPGA implementation of one module, every repair
window and all 114 phase-end outcomes agree with RTL simulation.

\section{Background and Problem Setting}
\label{sec:background}

\subsection{Target System Model}
\label{sec:target-system-model}

Morph-aware computing targets systems of computing modules that keep executing
and keep their local state while the system changes shape. We distinguish two
levels. The \emph{physical arrangement} is where the modules are and how they
behave, and it is not visible to the program. The \emph{topology} is the graph
\(T=(V,E)\) that the arrangement induces, where \(V\) is the set of modules and
\((u,w)\in E\) means that \(u\) and \(w\) can communicate directly. Formally,
a topology change is a committed transition from one such graph to
another. It may add or remove edges, split a connected component of \(T\), or
join two. A change in
arrangement matters to the program only when it changes \(T\).

Each module runs one persistent computation, a \emph{role}. The model
requires two things of a platform. A module keeps executing and keeps its local
state across a topology change. One read of its local link interface returns
its current one-hop \emph{neighbor snapshot} together with a revision that
advances on every committed change. Modules need not observe a change at the same time. State whose meaning depends
on \(T\), such as a cached next hop, is \emph{topology-dependent state}. It
becomes stale when a topology change makes the graph it described obsolete,
even though no write touched it, so staleness is relative to the revision, not
to a later write.

These requirements do not depend on the hardware. Shape-adaptive
chiplets~\cite{kadomoto2020riscv}, modular robots~\cite{goldstein2005matter,derosa2008ldp},
and the sensor and mesh nodes that run the stacks of
Section~\ref{sec:field-evidence} keep running while their neighbors change, and a
runtime can meet the second requirement in software by publishing a copy of its
neighbor table with a counter that advances on each change. Morph-aware
computing addresses the application state that survives a change, not the
routes that routing protocols rebuild~\cite{rfc6550,rfc7181}. We evaluate it
on MorphSim and on the RTL of one RISC-V chiplet (Section~\ref{sec:setup}),
but nothing in the model depends on chiplets.

\subsection{Evidence that the Problem Occurs in Practice}
\label{sec:field-evidence}

To test whether stale topology-dependent state causes failures in released
software, we searched the issue trackers of six network stacks (Contiki-NG,
Contiki, RIOT, OpenThread, Zephyr, and TinyOS). Twenty-eight symptom phrases
returned 103 unique reports, and we fixed this candidate set before reading it.
We included a report only if it named a topology-derived value, a topology or
membership change, evidence that the value survived the change, and a resulting
failure.

Thirteen reports from four stacks meet all four conditions. Ten are failed
sends and three are stale reads. We classified them by what morph-aware
computing could do about them. In six, a cached destination goes stale, the
pattern that send-target checking targets. In four more, the stack notifies the
module, but the failure comes from the order in which handlers run, which the
ordered-repair contract addresses. In the remaining three, the module receives no
notification at all. Appendix~\ref{app:field-evidence} lists all 13 with the
screening procedure.

The search shows that the problem occurs, but not how often. The authors chose the
repositories and phrases. Each search returned at most 40 results, and one
person made every inclusion judgment so that no rate can be estimated from the 103
candidates.

\subsection{Design Requirements}
\label{sec:design-requirements}

A program that handles topology changes correctly must do three things. It
must update its topology-dependent state after each change, keep all other
state, and make sure that no event handler uses the old state before the
update. A topology-change callback provides only the first. Unless something
orders the callback before every later \texttt{Send}, \texttt{Receive}, timer,
and device event, the program must compare revisions at each of these points
itself. We state these obligations, and what a compiler needs to enforce them,
as four requirements.

\emph{R1, explicit repair.}
The programming model provides a structured point for application-defined repair
of topology-dependent state.

\emph{R2, ordered repair.}
Once a role observes a newer topology revision, its repair must complete before
its next ordinary branch.

\emph{R3, persistent state.}
Application state survives a topology change unless explicitly updated by the program.

\emph{R4, compiler-controlled scheduling points.}
The execution structure exposes well-defined points at which compiler-generated
guards can observe topology changes and schedule repair.

Hand-written code can break R2 without any runtime errors, because it needs one
check per ordinary branch and nothing reports a missing one
(Section~\ref{sec:rq2}). MorphLang meets R2 through R4. All ordinary computation
lies in branches that the compiler sees, so MorphC guards every branch of a
role that declares repair. The repair body is a delimited region, so MorphC can
also analyze what it refreshes (Section~\ref{sec:send-target-checking}).
The requirements are local to one role. They need no globally synchronized
view of the topology; they leave the content of the repair to the application,
and they do not guarantee communication when no physical link exists.

\input{figures/fig-stack-compact}  

\section{Morph-Aware Programming Model and MorphLang}
\label{sec:morphlang}

This section shows how a programmer writes a role in MorphLang and how MorphC
compiles it. In a MorphLang program, the programmer states what to repair in a
\texttt{Reconfigure} branch (R1). At compilation time, MorphC decides when it
runs (R2), leaves all other state as it was (R3), and inserts its checks at the
branches (R4). Figure~\ref{fig:software-stack} follows the sensor of
Figure~\ref{fig:teaser} through send-target checking and code generation to
bare-metal RISC-V RV32I~\cite{riscv2024unpriv}.

\subsection{Roles and Events}
\label{sec:reactive-roles}

A MorphLang program defines roles, and MorphC compiles each role to its own
binary. A role is either a \texttt{coordinator}, which originates data, or a
\texttt{router}, which receives data and processes or forwards it. The
ordered-repair contract applies to both kinds of role.

A role reacts to events, and each event has its own branch. These
ordinary branches, inherited from MorphLang's earlier reactive
design~\cite{izawa2024reactive}, are \texttt{Send} and \texttt{Receive}, which
represent opportunities for outgoing and incoming communication, and
\texttt{Sample}, which represents a device-triggered action. A branch
communicates through the \texttt{Node} interface. For example, a \texttt{Send}
branch may call \texttt{Node.sendTo}, and a \texttt{Receive} branch may poll
incoming data through \texttt{Node.tryGetData}.

Morph-aware computing adds one optional branch, \texttt{Reconfigure}, which
holds the repair. Figure~\ref{fig:sensor-program} shows the sensor of
Figure~\ref{fig:teaser}. Line~2 declares \texttt{parent}, its
topology-dependent state. The \texttt{Send} branch (lines~4--5) sends to
\texttt{parent} without inspecting the topology, and the repair in
\texttt{Reconfigure} (lines~6--7) selects \texttt{parent} again from the
neighbor snapshot. MorphC places a guard before \texttt{Send}, so after the
replacement of Figure~\ref{fig:teaser} the repair runs before the next send.

\subsection{Core Syntax and Standard Interfaces}
\label{sec:morphlang-syntax}

\begin{figure}[t]
\begin{minipage}[h]{.33\linewidth}
\includegraphics{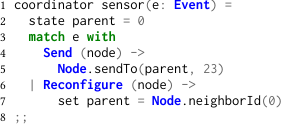}
\caption{The sensor of Figure~\ref{fig:teaser} in MorphLang.\label{fig:sensor-program}}
\Description{A MorphLang coordinator named sensor with state parent. Its Send
branch sends the value 23 to parent, and its Reconfigure branch sets parent to
the first neighbor in the snapshot.}
\end{minipage}\hfill
\begin{minipage}[h]{.64\linewidth}
{\footnotesize
\[
\begin{array}{@{}rcl@{}}
Prog &::=& \overline{t}\ \overline{f}\ \overline{\rho}\\
t &::=& \texttt{type}\ T=\{\overline{x:\sigma}\}
       \mid\texttt{type}\ T=C_1(\overline{\sigma})\mid\cdots\mid C_n(\overline{\sigma})\ \texttt{;;}\\
f &::=& \texttt{fun}\ g(\overline{x:\sigma})=a\ \texttt{;;}\\
\rho &::=& k\ x(e:\texttt{Event})=\overline{s}\ m\ \texttt{;;}\\
k &::=& \texttt{coordinator}\mid\texttt{router}\\
s &::=& \texttt{state}\ x=a
       \mid\texttt{device}\ x:\tau(\overline{\ell=a})\\
m &::=& \texttt{match}\ e\ \texttt{with}\ \overline{b}\\
b &::=& p\mathbin{\texttt{->}}a\\
p &::=& \texttt{Send(node)}\mid\texttt{Receive(node)}
       \mid\texttt{Reconfigure(node)}\\
a &::=& c\mid x\mid\texttt{let}\ x=a\ \texttt{in}\ a
       \mid\texttt{if}\ a\ \texttt{then}\ a\ \texttt{else}\ a\\
  &\mid& \texttt{set}\ x=a\mid a\mathbin{\mathit{op}}a
       \mid g(\overline{a})
       \mid x.f(\overline{a})\mid\texttt{Node}.f(\overline{a})\\
  &\mid& \texttt{match}\ a\ \texttt{with}\
       \overline{C(\overline{x})\mathbin{\texttt{->}}a}
\end{array}
\]
}
\caption{Core syntax of MorphLang, where \(\overline{x}\) denotes a sequence.\label{fig:syntax}}
\Description{A grammar in which a program consists of type, function, and
role definitions. A role is a coordinator or router with state and device
declarations and a match over Send, Receive, and Reconfigure events whose
bodies are expressions.}
\end{minipage}
\end{figure}

Figure~\ref{fig:syntax} summarizes the core language. A role declares
\texttt{state} variables, which persist across branch invocations and across
topology changes until the program changes them with \texttt{set}. A
\texttt{let} binding, in contrast, lasts for one branch invocation.
\texttt{Sample} branches, omitted from the grammar, are scheduled like other
ordinary branches.

A device expression \(x.f(\overline a)\) calls an interface introduced by a
\texttt{device} declaration. The predeclared \texttt{Node} interface sends to a
role or numeric core identifier, polls incoming data, and reads the current
neighbor snapshot. These are ordinary expressions inside branch bodies, so
MorphC schedules branches without inspecting their bodies.

The implementation includes further conventional features used by the
evaluation, including records and variants, non-recursive functions, mutable
references, arrays, bounded loops, and I2C/SPI device operations. These features
do not affect the semantics of topology changes. All values have statically known
size, allowing MorphC to generate code without garbage collection or a
general-purpose language runtime.

\section{Semantics and Implementation}
\label{sec:semantics-implementation}

\subsection{Operational Semantics of Topology-Change Handling}
\label{sec:operational-semantics}

This section fixes when a repair must run. The semantics ignores what a branch body does and when topology changes happen. For each role, it tracks only the latest revision the role has seen and whether a repair is still owed.

\subsubsection{State and Assumptions}

Publishing a topology change and seeing it are separate events. The execution
environment \emph{commits} a topology change by publishing a new neighbor
snapshot under the next revision. It then delivers that revision to each role on
its own schedule. Two roles may therefore see the same topology change at
different times.

For a role \(r\), let \(H(r,e)\) denote the branch body for event
\(e\). The event \(\mathsf{Reconfigure}\) is the repair, and every other event
is ordinary. The core
ordinary events are \(\mathsf{Send}\) and \(\mathsf{Receive}\), and an
implementation may add others such as \(\mathsf{Sample}\).
The semantics treats all ordinary events alike.

Running a branch is written as
\[
\mathsf{exec}(r,e,\delta,M)=M'
\]
and means that role \(r\) runs branch \(e\) on its delivered
snapshot \(\delta\), so the runtime state \(M\) becomes \(M'\). The runtime
state holds persistent role state, queries, and observable output. This lets the
rules fix the order of branches without describing MorphLang expressions.

Let
\[
\mathcal R=\{r\mid\mathsf{Reconfigure}\in\operatorname{dom}(H(r,\cdot))\}
\]
be the roles that declare \texttt{Reconfigure}. The state of the model is
\[
\langle v,T,d,\Delta,M,P\rangle ,
\]
where \(v\) is the latest committed revision and \(T\) is its neighbor snapshot. The
map \(d\) gives the revision delivered to each role, and \(\Delta(r)\) is the
snapshot delivered to that role. The set \(P\subseteq\mathcal R\) holds the roles that still owe a
\texttt{Reconfigure} run. Execution starts with \(v=0\), \(d(r)=0\),
\(\Delta(r)=T\) for every \(r\), and \(P=\emptyset\). Each role starts from its declared initial state under the initial snapshot
\(T\), which does not trigger \texttt{Reconfigure}. The contract
covers revisions delivered after a role starts, and keeping the initial state
consistent with the initial snapshot is the programmer's obligation.

The model keeps the revision \(v\) apart from the
snapshot \(T\) because equal snapshots are not enough. If the topology
changes from \(T_0\) to \(T_1\) and back to \(T_0\), the two snapshots are
equal, yet the role must still learn that a newer revision exists.
Each branch run is one step, so in this model a commit falls between two branch
runs. Section~\ref{sec:morphc} relaxes this.

\begin{samepage}
\subsubsection{Operational Rules}

Four rules say which branch may run.

\[
\frac{\ }
{\langle v,T,d,\Delta,M,P\rangle
\longrightarrow
\langle v+1,T',d,\Delta,M,P\rangle}
\quad\textsc{Commit}
\]

\[
\frac{d(r)<v}
{\langle v,T,d,\Delta,M,P\rangle
\longrightarrow
\langle v,T,d[r\mapsto v],\Delta[r\mapsto T],M,
P\cup(\{r\}\cap\mathcal R)\rangle}
\quad\textsc{Deliver}
\]

\[
\frac{r\in P\qquad
\mathsf{exec}(r,\mathsf{Reconfigure},\Delta(r),M)=M'}
{\langle v,T,d,\Delta,M,P\rangle\xrightarrow{r:\mathsf{reconfigure}}
\langle v,T,d,\Delta,M',P\setminus\{r\}\rangle}
\quad\textsc{Reconfigure}
\]

\[
\frac{r\notin P\qquad
e\in\operatorname{dom}(H(r,\cdot))\setminus\{\mathsf{Reconfigure}\}
\qquad
\mathsf{exec}(r,e,\Delta(r),M)=M'}
{\langle v,T,d,\Delta,M,P\rangle\xrightarrow{r:e}
\langle v,T,d,\Delta,M',P\rangle}
\quad\textsc{Ordinary}
\]
\end{samepage}

\textsc{Commit} publishes a snapshot and leaves application state alone.
\textsc{Deliver} gives one role the latest committed snapshot and, if the role
declares \texttt{Reconfigure}, marks it as owing a repair. \textsc{Reconfigure}
runs the repair on the delivered snapshot and clears the mark. \textsc{Ordinary}
applies only while the role has no mark.

In the example, write $S$ for the sensor. It has received revision~4 and run
its repair for it, so $S\notin P$, and it stores relay~3 in \texttt{parent}. The
replacement commits revision~5, and delivering it to $S$ marks $S$ as owing a
repair.
\[
S\notin P \xrightarrow{\textsc{Commit}(5)} S\notin P
\xrightarrow{\textsc{Deliver}(S,5)} S\in P.
\]
While $S\in P$, \textsc{Ordinary} does not apply to $S$, so its \texttt{Send}
cannot run. \textsc{Reconfigure} runs the repair on the delivered
snapshot and removes $S$ from $P$.

A delivery always installs the newest committed revision. The rules otherwise leave the scheduling among
roles and among ordinary events free.
Section~\ref{sec:morphc} gives the schedule MorphC generates.

\subsubsection{Properties of the Execution Contract}
\label{sec:execution-contract}

The rules give three guarantees, sketched here and mechanized in
Section~\ref{sec:dispatch-validation}. Theorem~\ref{thm:ordered-repair} is the
ordered-repair contract itself (requirement R2 of
Section~\ref{sec:design-requirements}), so a role never runs an ordinary branch
on a snapshot it has not yet repaired for. Theorem~\ref{thm:coalescing} shows that
changes arriving faster than a role repairs need only one repair, and
Theorem~\ref{thm:continuity} concerns the application state (requirement R3 of
Section~\ref{sec:design-requirements}), which a
topology change by itself never touches.

\begin{theorem}[Ordered repair]\label{thm:ordered-repair}
If a role \(r\in\mathcal R\) receives a newer revision and later runs an
ordinary branch, an \(r:\mathsf{reconfigure}\) step occurs after the most recent
such delivery and before that branch.
\end{theorem}

\begin{proof}[Proof sketch]
\textsc{Deliver} puts \(r\) in \(P\), \textsc{Ordinary} does not apply while
\(r\in P\), and only \textsc{Reconfigure} removes \(r\) from \(P\).
\end{proof}

\begin{theorem}[Coalescing]\label{thm:coalescing}
Each delivered revision leads to at most one \(r:\mathsf{reconfigure}\) step.
If \(k\geq 1\) commits occur before a role \(r\in\mathcal R\) receives a
delivery, that delivery sets \(d(r)\) to the newest revision \(v\), and none of
the intermediate revisions is ever repaired.
\end{theorem}

\begin{proof}[Proof sketch]
\textsc{Reconfigure} removes \(r\) from \(P\), and only a \textsc{Deliver}
of a strictly newer revision, which requires \(d(r)<v\), adds it back, so a
delivered revision is repaired at most once. \textsc{Deliver} always installs
the newest committed revision, so after \(k\) commits it sets \(d(r)=v\)
directly. Only a revision in \(d(r)\) can be repaired, and no step installs an
intermediate revision there, so none is repaired.
\end{proof}

\begin{theorem}[State continuity]\label{thm:continuity}
\textsc{Commit} and \textsc{Deliver} leave \(M\) unchanged.
\end{theorem}

\begin{proof}[Proof sketch]
Immediate from the two rules, whose conclusions carry \(M\) over unchanged.
\end{proof}

\noindent Theorem~\ref{thm:coalescing} is what allows an
implementation to keep one revision number per role instead of a log of
changes.

The theorems fix only when the repair runs, not what it achieves. The
following proposition, which is not mechanized, connects the two. It concerns a property
\(\Phi_r(T,M)\) that a role \(r\in\mathcal R\) should keep, a condition on a
snapshot \(T\) and the state \(M\) whose content the application decides. For
the sensor of Figure~\ref{fig:teaser}, \(\Phi_r\) says that \texttt{parent} is
a neighbor in \(T\). The proposition states that if the program keeps
\(\Phi_r\) in three ways, every ordinary branch of \(r\) sees it hold.

\begin{proposition}[Conditional application invariant]\label{prop:invariant}
Suppose that
(1)~\(\Phi_r(\Delta(r),M)\) holds when execution starts,
(2)~every \texttt{Reconfigure} run of \(r\) makes \(\Phi_r\) true for the snapshot
\(\Delta(r)\) it runs on, whatever \(M\) is, and
(3)~every other step, such as an ordinary branch of \(r\) or a step of another
role, keeps \(\Phi_r\) true while \(\Delta(r)\) is unchanged.
Then \(\Phi_r(\Delta(r),M)\) holds whenever \(r\) starts an ordinary branch.
\end{proposition}

\begin{proof}[Proof sketch]
Only a \textsc{Deliver} to \(r\) changes \(\Delta(r)\),
the snapshot the role holds. We show by induction over steps that
\(\Phi_r(\Delta(r),M)\) holds whenever \(r\notin P\). It holds initially
by~(1), since \(P=\emptyset\). A \textsc{Deliver} to \(r\) may change the
snapshot and break \(\Phi_r\), but it puts \(r\) in \(P\) at the same time, and
while \(r\in P\) the \textsc{Ordinary} rule does not apply to \(r\). The next
\textsc{Reconfigure} of \(r\) makes \(\Phi_r\) true again by~(2) and removes
\(r\) from \(P\). Every other step leaves \(\Delta(r)\) unchanged and keeps
\(\Phi_r\) by~(3). An ordinary branch of \(r\) therefore starts only in a state
with \(r\notin P\), where \(\Phi_r\) holds.
\end{proof}

\noindent Condition~(2) must hold
for any \(M\), because other roles may change \(M\) while \(r\) waits for its
repair. The snapshot in~(3) is the delivered \(\Delta(r)\), not the committed
\(T\), which \textsc{Commit} may change at any time.
Theorem~\ref{thm:check-sound} proves the corresponding property for send
targets, assuming the repair has run (Theorem~\ref{thm:ordered-repair}).

\subsection{Compiler Implementation of Ordered Repair}
\label{sec:morphc}

MorphC implements the rules of Section~\ref{sec:operational-semantics} as
generated code.

For each role that declares \texttt{Reconfigure}, MorphC keeps the
\emph{handled revision}, the last revision whose repair has finished. Initially,
it is the revision under which the role starts. Before each ordinary branch, a
generated guard reads the revision reported by the local link interface and
compares it with the handled revision. If they differ, a repair helper reads
the revision again, runs every repair body on the snapshot from that read, stores that
revision as handled only after all repair code returns,
and then enters the branch. This store of the handled revision records a finished
repair, not entry into the repair body. A role without \texttt{Reconfigure}
owes no repair and gets no guards.

The generated code is one loop that runs every \texttt{Send} and
\texttt{Receive} branch in source order and then every \texttt{Sample}
branch. Each branch has its own guard, placed before the branch's own gate,
such as the test for a pending message, so a repair can run even when the
branch body is then skipped. Section~\ref{sec:rq2} compares the generated guards with manually written
checks.

\paragraph{How the code matches the rules.}
When a guard finds that the delivered and handled revisions differ, a topology
change has been delivered, and its repair has not finished. This case
corresponds to \(r\in P\) in Section~\ref{sec:operational-semantics}. The guard lets an ordinary
branch start only after the repair has finished, which is the condition
\(r\notin P\) of the \textsc{Ordinary} rule.

\paragraph{Reading the revision and snapshot together.}
Delivery happens at points the compiler controls. A commit may
arrive while a branch runs, but the core sees it only at its next revision
read. One read of the link interface returns the latest committed revision and
its neighbor snapshot together. Such reads occur only in guards and in the repair
helper, so delivery happens only there. This is one of the orders that the
\textsc{Deliver} rule allows. MorphSim and the RTL of Section~\ref{sec:setup}
hold that snapshot fixed until the next read, through the repair and the branch it enters.
That guarantee comes from the single read in the link interface, which direct reads of a
changing table would not give.

\paragraph{Dropped sends.}
A send to a neighbor that has left the snapshot the core holds is lost, a
\emph{dropped send}, and is not redirected to a module that took the same table
slot. A consistent snapshot therefore does not prevent dropped sends before the repair.

\paragraph{Revision width.}
The revision is a 32-bit word compared only for inequality, so correctness assumes that
a role never falls behind by a positive multiple of \(2^{32}\) revisions. The Rocq models
use unbounded naturals, the TLA+ models bound revisions, and neither models wraparound.

Other constructs compile conventionally. A \texttt{state} declaration becomes role-local
memory, \texttt{let} bindings become temporaries, and \texttt{Node}
operations become memory-mapped I/O.

\subsection{Send-Target Checking}
\label{sec:send-target-checking}
\label{sec:state-classification}

Unlike application-specific state such as an accumulator, numeric send targets
have a link to the topology that is often visible in the code, for example, through
\texttt{Node.neighborId()}. Send-target checking rejects two forms of
stale reuse that the code shows. One is a destination fixed independently of the
current topology. The other is a persistent destination that survives a topology
change without being refreshed on every repair path for the delivered revision. The check does not
prove that an accepted destination is currently adjacent.
Figure~\ref{fig:software-stack}(c) applies the check to three
variants of the sensor.

For this check, MorphC tracks where each numeric expression may get its value. These
\emph{origins} fall in three categories. A \(\mathsf{current}\) value
comes from the neighbor snapshot of the role's delivered revision.
A \(\mathsf{dynamic}\) value comes at runtime from a source whose
relation to the topology is unknown, such as a received message or the
sender of the last message. A \(\mathsf{fixed}\) value does not change with the
topology. It is either a literal core identifier or persistent state that is not refreshed
on every repair path for the delivered revision. The category names what the analysis
could not establish, not what is known to be stale. A \(\mathsf{fixed}\)
destination may still be adjacent at runtime.

The analysis records the \emph{set} of origins that may reach an expression.
The judgment \(\Gamma;G\vdash a:q\), with
\(q\subseteq\{\mathsf{current},\mathsf{dynamic},\mathsf{fixed}\}\), states that
the analysis assigns \(a\) the origin set \(q\). Here \(\Gamma\) records the origin sets of
local values and \(G\) the persistent locations that the analysis can prove are
refreshed for the delivered revision in \texttt{Reconfigure}. Where control flow
merges, and for expressions that combine several operands, the origin sets are combined by
union, so a value is never treated as current merely because one incoming path
took it from the snapshot. Thus \texttt{Node.neighborId(0)} has origin
\(\{\mathsf{current}\}\), a received identifier \(\{\mathsf{dynamic}\}\), and
both the literal \texttt{7} and an unrefreshed persistent destination
\(\{\mathsf{fixed}\}\).

A persistent location refreshed on every
\texttt{Reconfigure} path with a value whose only origin is \(\mathsf{current}\)
contributes \(\mathsf{current}\) afterwards. Otherwise it contributes
\(\mathsf{fixed}\). An ordinary branch may overwrite such a location before a
later send. At that send, the location contributes \(\mathsf{current}\) together with the
origins of every value an ordinary branch may store in it. MorphC also warns
when a send target may read a location that such an overwrite may set to a \(\mathsf{dynamic}\) value.

MorphC accepts \texttt{Node.sendTo(target,\ldots)} exactly
when \(\mathsf{fixed}\notin q_{\mathit{target}}\), and a rejected send is a
compile error.

MorphC accepts a \(\mathsf{dynamic}\) target, whose origin does not establish
adjacency, because arbitrary runtime values lie outside the stale-reuse
property checked here. A role name, such as a sink's, is treated the same
way. It says that the destination is an endpoint chosen at deployment rather
than a cached neighbor, whereas a literal core identifier says no such thing
and is \(\mathsf{fixed}\). Acceptance therefore means only that MorphC found no
possible \(\mathsf{fixed}\) origin.
An evaluation-only compiler option turns the check off so that rejected
programs can still be run.

\paragraph{Why the check depends on ordered repair.}
An accepted send is safe only if the repair has already run when the branch sends. Consider the
sensor of Figure~\ref{fig:teaser} with \texttt{Reconfigure (n) -> set parent
= Node.neighborId(0)} and \texttt{Send (n) -> Node.sendTo(parent, v)}. Every
repair path refreshes \texttt{parent} from the snapshot, so the target is
\(\{\mathsf{current}\}\) and accepted. If a \texttt{Send} can run after a
topology change but before the repair, as with a notification callback that
the runtime queues behind already scheduled work, or a hand-written loop that
omits the check before \texttt{Send}, the same accepted send reaches the
departed relay. The conclusion ``\texttt{parent} is current'' is true only
after the repair. MorphC's guards make it true at every send, and
removing one guard from the generated code turns exactly this case into sends to
departed relays (Section~\ref{sec:rq2}). For the same reason, the soundness results below assume
Theorem~\ref{thm:ordered-repair}.

Let a persistent location be \emph{fresh} when its value was last written from the
neighbor snapshot of the role's delivered revision. For the core language of the check, we prove the
following two theorems, the second extending the first to ordinary branches that
run between the repair and the send (Section~\ref{sec:send-target-validation}).

\begin{theorem}[Soundness of send-target checking]\label{thm:check-sound}
  (a)~If a destination contains a literal, a local whose origin may be
  \(\mathsf{fixed}\), or a persistent location outside \(G\), the check rejects
  it.
  (b)~Suppose the analysis of \texttt{Reconfigure} starts from an empty \(G\) and
  the repair runs to completion at the delivered revision. Then every location in
  \(G\) is fresh, and every destination built only from neighbor-snapshot reads
  and locations in \(G\) is accepted and fresh.
\end{theorem}

\begin{theorem}[Preservation]\label{thm:preservation}
Let \(O\) record the origins of the values that ordinary branches may store
into locations in \(G\), where a store guarded by an equality test against a
current expression counts as current, and suppose \(O\) covers every such store. After the
repair and any finite sequence of ordinary branches at the same delivered
revision, including a prefix of the sending branch, the following hold.
(a)~Every destination the check accepts has no \(\mathsf{fixed}\) origin in the
run-time origin tags of the model.
(b)~A destination built only from neighbor-snapshot reads and locations in
\(G\) is still accepted and fresh if \(O\) records only \(\mathsf{current}\)
values for the locations it reads.
\end{theorem}

The check covers revisions delivered after a role starts. A role that
starts with neighbors in range runs on the initial values of its \texttt{state}
declarations until its next topology change, and the check does not examine
those values, so Figure~\ref{fig:software-stack}'s \texttt{state parent = 0} is
accepted. The evaluation starts every role with no neighbor in range
(Section~\ref{sec:setup}), so each first attachment is a revision change that
runs \texttt{Reconfigure}. A runtime can remove the obligation without
changing the contract by committing a first revision before it starts the
cores, so that the first guard of each role runs \texttt{Reconfigure} on the
initial snapshot. MorphSim and the chiplet RTL do not yet do so.

\section{Mechanization and Validation}
\label{sec:validation}

This section describes how we established the results of
Section~\ref{sec:semantics-implementation} and how we connect them to MorphC.
Rocq~\cite{coq} proves Theorems~\ref{thm:ordered-repair}, \ref{thm:check-sound},
and~\ref{thm:preservation} and the parts of Theorems~\ref{thm:coalescing}
and~\ref{thm:continuity} stated below, in about 2{,}400 lines.
Appendices~\ref{app:tla-detail} and~\ref{app:coq-send-target} give the
definitions and lemma names. The proofs are about models, not about MorphC's
Rust source, so for each guarantee we also check the compiler's output and
state what remains trusted. Appendix Table~\ref{tab:claim-evidence} maps each
claim of this section and of Section~\ref{sec:evaluation} to its evidence.

\subsection{Ordered Repair}
\label{sec:dispatch-validation}

\paragraph{Mechanization.}
The Rocq model follows the rules of Section~\ref{sec:operational-semantics},
with the guard applied at every ordinary-branch entry in the same step as that
entry, so the repair of Theorem~\ref{thm:ordered-repair} precedes the branch
within one transition. It bounds neither the number of roles nor the number of commits.
Theorem~\ref{thm:ordered-repair} (\nolinkurl{ordinary_runs_after_repair}) and
the first sentence of Theorem~\ref{thm:coalescing}
(\nolinkurl{reconfigure_at_most_once}) are proved from invariants of the
reachable states. For the second sentence, Rocq proves that an execution with
one delivery and one repair brings \(r\) to \(v\) and records exactly one repair,
for revision \(v\) (\nolinkurl{overlapping_commits_coalesce}). That no execution repairs
an intermediate revision follows from the rules as sketched above and is not
mechanized. Theorem~\ref{thm:continuity}
(\nolinkurl{topology_change_preserves_application_state}) holds for every
single commit and delivery step, where the model's \(M\) is the role's state.
The model also records the handled revision
of Section~\ref{sec:morphc} and proves that it never exceeds the delivered one
(\nolinkurl{acknowledgement_not_ahead_of_delivery}).

A second Rocq model splits the guard, the repair, and the branch entry into
separate steps and lets a commit arrive while a repair or a branch runs. For one
role, it proves that the delivered revision stays fixed during the repair and
that the role enters an ordinary branch only after that repair has completed.
That the snapshot also stays fixed follows from the link interface of
Section~\ref{sec:morphc}, not from the model.

We also wrote the rules in TLA+.
TLC explores all interleavings of commits, deliveries, message arrivals, and
scheduling in six configurations of three to five roles and finds no violation.
It also checks under weak fairness that a pending repair eventually runs, and it finds a
counterexample for each of five variants that each break one part of the
guard (Appendix~\ref{app:tla-detail}).

\paragraph{Checking the implementation.}
The theorems apply to MorphC only if its output has the structure of the model.
A checker in our evaluation harness decodes every RV32I binary and verifies
that an ordinary branch is entered only when a freshly read revision equals the
handled revision, that a different revision leads to the repair first, and that
the handled revision is stored only after that path enters the repair and
stores to every state cell the repair assigns. All 25 role
binaries of the evaluation pass, 17 of which declare a repair. They are the 19
programs, the three roles of the two-hop application, and three
guard-benchmark builds (the two microbenchmark versions and an instrumented
rebuild of the guarded one). To test the checker and our measurements, we
inject four kinds of bugs into the code generator (Table~\ref{tab:mutation}).
They omit a guard, move a guard after its branch, record a revision as handled
without running the repair, or read the revision once per loop instead of
before each branch. The first two match the missing-check and stale-target TLA+ models, and the
last resembles the model that reads the global revision counter. The checker
rejects all 68 resulting binaries. Because a later guard runs an omitted repair
within tens of cycles, we also measure how long a role keeps communicating
after a revision is committed but before its repair is recorded
(Appendix~\ref{app:dispatch-detail}). By this measure or by a change in state,
66 of the 68 binaries behave differently from the correct binary. Finally,
triggering a relay replacement at each of the 548 cycles between the start of
a repair and the next ordinary branch of one program produces no violation in
81{,}652 branch executions, and Section~\ref{sec:fpga} repeats such commits on
an FPGA.

\begin{table}[t]
\caption{Injected guard bugs in the 13 repair-declaring programs on four schedules. A
mutant differs \emph{by interval} if it communicates under an unhandled revision
in more phases or for longer, and \emph{by state} if a handled revision, change
indicator, or destination validity differs. The two criteria overlap (three mutants
differ by both), so the columns do not sum to 66. Both survivors place the guard after its branch.}
\label{tab:mutation}
\scriptsize
\begin{tabular}{lrrrr}
\toprule
Mutation operator & Mutants & By interval & By state & Survive \\
\midrule
Omit one guard site & 29 & 29 & 0 & 0 \\
Guard after its branch & 13 & 11 & 0 & 2 \\
Record handled revision without repair & 13 & 6 & 10 & 0 \\
Read the revision once per loop & 13 & 13 & 0 & 0 \\
\midrule
Total & 68 & 59 & 10 & 2 \\
\bottomrule
\end{tabular}
\end{table}

\paragraph{Trusted components.}
The checker validates each binary, not the compiler. It relies on the addresses of
the branch entries, the dispatch loop, the
handled revision, and the repair entry that MorphC reports, on the list of
state cells the repair assigns in the source, and assumes that array stores do not alias
the handled revision or the stack and that callees preserve callee-saved
registers. The Rocq model runs all cases of one kind
(\texttt{Send}, \texttt{Receive}, or \texttt{Sample}) under one guard, whereas
MorphC guards each case, so a delivery between two cases of the same kind falls
before or after the whole step in the model. Duplicate repair and a handled
revision recorded before delivery are covered only by the TLA+ models, not by
the injected bugs.

\subsection{Send-Target Checking}
\label{sec:send-target-validation}

\paragraph{Mechanization.}
We mechanize the origin rules of Section~\ref{sec:send-target-checking} and the
refresh analysis for a core language with assignments, calls, conditionals, and
loops in the repair, and without calls in ordinary branches. The Rocq execution
labels each persistent location with the revision of
the snapshot from which its value was written from, which gives the notion of freshness in
Theorems~\ref{thm:check-sound} and~\ref{thm:preservation}. In the analysis, an
assignment from the neighbor snapshot records a location in \(G\), a later
assignment from another source removes it, a call removes the locations it may
write, a conditional keeps only locations that both branches record, and a loop
keeps only locations its body preserves. A loop that runs at
least once whenever the role has a neighbor, such as
\texttt{for i in 0..Node.neighborCount()}, may also record
a refresh from its first iteration. The proof of
Theorem~\ref{thm:check-sound} (\nolinkurl{An_sound}) establishes that these
rules record only fresh locations, and
Theorem~\ref{thm:preservation}
(\nolinkurl{reconfigure_then_ordinary_send_is_fresh}) covers the ordinary
branches that follow. Rocq examples check that the rules do not record a
location refreshed on only one branch of a conditional, only inside a loop that
may run zero times, or overwritten afterward.

\paragraph{Checking the implementation.}
For the origin rules, MorphC's decision equals that of a hand translation of
the Rocq check on all 3{,}143 core expressions of depth at most three under one
fixed environment. This test covers neither the refresh analysis nor the
ordinary-branch extension of Theorem~\ref{thm:preservation}.

\paragraph{Trusted components.}
MorphC's Rust implementation is not proved to follow the Rocq rules, and we do
not prove that it computes the complete set of stored values that
Theorem~\ref{thm:preservation} assumes. The origins of local values are given
to the model, not computed. Records, arrays, and references are modeled as persistent locations,
and both the analysis and the model assume that a condition or loop bound whose value is
fixed whenever the role has a neighbor always takes the path that value
selects. Whether a destination obtained at runtime is a neighbor is outside the
theorems. No Rocq theorem connects this model to that of ordered repair, so
the use of Theorem~\ref{thm:ordered-repair} here is a paper argument.

\section{Evaluation}
\label{sec:evaluation}

This section answers four research questions. RQ1 to RQ3, run in the MorphSim
simulator, examine the effect of repair, the coordination and checking that
MorphC supplies, and the execution cost of the generated code. RQ4 checks
whether these results hold on an FPGA.

We evaluate 19 MorphLang programs that we wrote by hand. Thirteen re-implement
the destination-selection logic of a routine from Contiki-NG, TinyOS,
OpenThread, Zephyr, RIOT, Meshtastic, or OONF's OLSRv2 implementation,
including RPL and Trickle routines
\cite{oikonomou2022contiking,levis2005tinyos,openthread,zephyr,baccelli2018riot,meshtastic,rfc7181,rfc6550,rfc6206}.
Each keeps the routine's cached destination, the events that update it, and its
assignments and conditionals in upstream order, and replaces the surrounding
protocol and framework code with simplified state.


\subsection{Research Questions}

Section~\ref{sec:validation} established that every repair runs before the next
ordinary branch for a model of the generated code and checked it on the generated
binaries. The evaluation asks what this buys a program
in practice.

\begin{description}
\item [RQ1] How does application-defined repair affect behavior after topology
  changes?
\item [RQ2] What do MorphC's revision guards and send-target checking add to
  the programmer's repair?
\item [RQ3] What are the execution and code-size costs of the MorphLang/MorphC
  implementation?
\item [RQ4] Do the simulated outcomes and costs hold on an FPGA?
\end{description}

\subsection{Experimental Setup}
\label{sec:setup}
\label{sec:execution-infrastructure}

\paragraph{Platforms.}
\label{sec:lower-layers}
RQ1 to RQ3 run in MorphSim, a cycle-level simulator that executes MorphC's
bare-metal RV32I binaries, with one simulated core per chiplet. It
implements the topology, local memory, communication queues, and revision
interface of Section~\ref{sec:operational-semantics}, and derives neighbor
snapshots from chiplet positions. At fixed cycle counts, a host-controlled
change rebuilds the snapshots and advances the committed revision. Such a
commit may land inside a branch or a guard, and the core sees it only at its
next revision read. Each core's receive queue holds 16 words, as in the
chiplet's hardware, and a word arriving at a full queue is lost. Execution is
deterministic, so every schedule replays exactly. MorphSim's timing is not
calibrated against hardware. We check it against a packet-level reference
model, the RTL of the chiplet in simulation, and QEMU, which establishes
consistency between these simulators but not physical timing
(Appendix~\ref{app:cross-layer}).
The RTL\anon[]{ of mm-node} implements a node in SystemVerilog, with a two-stage
in-order RV32I core and a link peripheral that exposes the neighbor snapshot, a
16-word receive queue, and the revision as memory-mapped registers. One aligned
load of the revision register returns the revision together with its neighbor
table, which is the read of Section~\ref{sec:morphc}. A host port commits new
neighbor tables and logs every send, including dropped ones. We simulate the RTL with Icarus
Verilog, and RQ4 synthesizes the same sources for an FPGA
(Section~\ref{sec:fpga}).

\input{tables/tab-programs}

\paragraph{Programs.}
We selected functions and handlers from released network stacks that cache a peer,
parent, or next hop, and recorded the cached value, its selection code, its
re-selection call sites, and its refresh trigger at a fixed upstream commit.
This purposive set exercises cached-destination behavior. It is not a random
sample. Table~\ref{tab:oss-units} lists the 13 routines, their update triggers,
the derived variants, and the C comparisons. Where upstream
delays an update through a callback, \texttt{Reconfigure} runs it immediately
(Section~\ref{sec:threats}). Appendix~\ref{app:oss-translation-map} lists, per
routine, the preserved decision and the omitted context. The artifact retrieves
all 13 upstream routines at the recorded commits and verifies their hashes and
line locations.

\paragraph{Derived variants.}
The six derived variants are of two kinds. Programs~6, 16, and~17 keep the
update logic and change only when it runs. Programs~6 and~16 move an existing
update body, including its conditions, into \texttt{Reconfigure}, and
Program~17 keeps its timer path and also calls the same candidate scan from
\texttt{Reconfigure}. Programs~14, 18, and~19 come from routines that learn a
destination from received messages or only invalidate it, so there is no upstream
re-selection body to move. They add a \texttt{Reconfigure} branch that selects
from the current neighbor snapshot and leave every other branch unchanged.
Section~\ref{sec:rq1-trigger} explains why the distinction matters.

\paragraph{C baselines.}
For RQ3, we build C executables with the same state, trigger, and
simplifications. Of the nine C baselines, six compile unchanged upstream C function bodies, with small
replacements for framework dependencies. For three nesC or
framework-dependent C/C++ routines, we write C versions by hand and label them
\emph{manual translations}. Five of the eight baselines that meet the matching
criterion below compile byte-exact upstream extracts. The six derived variants have no upstream
counterpart, and we did not build C executables for four further programs.
We report a cost ratio only when both versions make the same routing
decisions. After every topology change, they must store the same destination,
agree on whether it is a current neighbor, change it in the same phases, and
miss a current neighbor in the same phases. We do not require equal delivery or
transmission counts, because the work completed within a fixed cycle budget is
what RQ3 measures.

\paragraph{Topology schedules.}
Every per-program count uses one reference schedule. A pool of three relays
surrounds the chiplet under test. After one initial phase with all relays out of
range, the schedule attaches one relay and then makes five replacements that
keep the neighbor count, each phase lasting 90{,}000 cycles, with seed~0. The
MorphLang and C versions of a program use the same schedule. RQ1 also runs 60
schedules per program variant, namely 20 replacement sequences at three phase
lengths, one of which is the reference schedule. Several factors vary together
across these runs, so they test whether the direction of a paired difference
persists, not the effect of each factor.

\subsection{RQ1 (Effect of Application-Defined Repair)}
\label{sec:rq1}

RQ1 asks whether a program is worse off without its repair. Removing repair
lowers destination validity in ten of the 13 programs that declare one and
never raises it (Table~\ref{tab:rq1-exact}). All seven programs whose repair
comes from upstream code get worse, and six of them go from no phase to all
six phases ending at a non-neighbor.


\paragraph{Method.}
The 13 programs that declare \texttt{Reconfigure} are the seven whose repair
comes from upstream code (Programs~1--3, 7, 9, 12, and~15) and the six derived
variants (Programs~6, 14, and~16--19, Table~\ref{tab:oss-units}). For each, we
compare the program as written with two ablations.
\begin{itemize}
\item \emph{Branch removed.} The whole \texttt{Reconfigure} branch is deleted.
The role then owes no repair, so MorphC also omits its guards. This measures
the effect of declaring repair, including the guards' cost.
\item \emph{Body emptied.} The repair body is replaced with a unit, and the guards
and the handled revision stay. This keeps the guards, although removing
body instructions still changes timing.
\end{itemize}
At each phase end, a phase is a \emph{non-neighbor phase} if at least one relay
is adjacent but the cached destination is not one of them. \emph{Destination
validity} is its complement, so fewer non-neighbor phases is better. We count
phases rather than dropped sends, because MorphC's runtime issues no
transmission to a destination that matches no neighbor.

\paragraph{Result.}

\input{tables/tab-rq1-exact}
\input{tables/tab-rq1-sequence}

Of the seven programs whose repair comes from upstream,
\texttt{openthread-leader} worsens least, from 0/6 to 4/6, because a comparison
in its \texttt{Receive} branch redirects some later sends even without
\texttt{Reconfigure}. Of the six derived variants, three worsen. The other
three, \texttt{meshtastic-nexthop-safe}, \texttt{zephyr-nud-\linebreak[0]reactive},
and \texttt{zephyr-friend-\linebreak[0]reactive}, stay at 0/6 either way,
because a \texttt{Receive} branch also sets the destination to the sender of the
last message, so the messages from the new relay do the repair's work
(Table~\ref{tab:rq1-exact}). The result holds across 20 replacement
sequences that vary the seed, relay-pool size, neighbor degree, and replaced
relay, each at phase lengths of 45{,}000, 90{,}000, and 180{,}000 cycles.
Removing repair never improves destination validity; eight programs are worse
in every run and two (Programs~9 and~17) in most; the three that learn from
received messages are never affected, and
emptying only the body gives the same classification
(Table~\ref{tab:rq1-sequences}).

\paragraph{Derived variants against their originals.}
\label{sec:rq1-trigger}Comparing each derived variant with its original shows
the effect of running an update on a topology change. Programs~6 (from~5) and~17 (from~8) improve.
Program~16 (from~4) stays at 1/6, because \texttt{contiki-rplclassic} selects a new
parent only when earlier feedback has set a flag, and the variant keeps this
condition. The three variants that add a repair are unchanged, because their
originals already learn the destination from received messages.

\paragraph{Result beyond a single cached destination.}
\label{sec:whole-app}\label{sec:aggregator}
Two further programs test repair across two hops and for states other than a destination.
In a sensor-relay-sink application, the sensor and the relay each cache a next
hop and declare a repair. Over eight phases, the relay is replaced from a pool
of three, the sink detaches once, and the relay detaches once. As written,
readings cross both hops in all six phases with a complete chain, 6{,}822 in
total, and deleting either repair yields no deliveries. A variant that learns
each hop only once, on its first nonempty snapshot, delivers only in the two
phases with the original relay.

The second program is an aggregator in the
style of TAG~\cite{madden2002tag} that keeps one reading per neighbor and sends
their sum to a sink. Its repair evicts the readings of senders that are no
longer neighbors, and since each source reports a distinct power of two, the
sum shows exactly which sources it contains. Over 20 schedules of 12 phases
that attach, add, replace, and remove sources, the sum is exact at the end of
all 240 phases as written. Without the repair, it includes a departed source
in 209 phases, including all 144 that follow a departure. Eviction alone is not
enough, because readings a departed source sent before the change are still
queued. The program therefore also accepts a reading only from a current
neighbor, and without that check nine phases still include a departed source.

\paragraph{Answer.}
Repair matters wherever no other code keeps the destination current (ten of 13
programs), and it changes nothing for the three that learn it from received messages. The
two further programs show the same dependence on repair, but the contract fixes
only when repair runs, not whether later branches preserve what it repaired.

\subsection{RQ2 (Revision Guards and Send-Target Checking)}
\label{sec:rq2}

RQ2 asks what MorphC adds beyond the repair the programmer writes, namely the
revision guards and send-target checking.

\paragraph{Method.}
We decode each RV32I binary and count the guards before ordinary branches. For
five of the programs we also wrote C versions whose loop checks the
revision before every ordinary branch and records it only after the repair
returns. To see what a single missing guard costs, we reuse the injected bugs
of Section~\ref{sec:dispatch-validation} and count, on the four schedules of
Table~\ref{tab:mutation}, the sends dropped at a departed relay and the sends
issued under an unhandled revision. For send-target checking, we compile each
of the 19 programs and record whether MorphC accepts every
\texttt{Node.sendTo} destination.

\paragraph{Result.}
MorphC emits 29 guards, one per ordinary branch of the 13 repair-declaring
programs, and none for the six that declare no \texttt{Reconfigure}. The five C
loops contain 11 checks at the same places as MorphC's 11 guards, and all five
meet the matching criterion of Section~\ref{sec:setup}. The C compiler, however,
accepts the loop with any check missing. Removing the first guard of each program
raises the dropped sends from
82 to 165 and the sends under an unhandled revision from 108 to 213 (totals over
the 13 programs and four schedules), with an increase in every program. Over all
29 single-guard omissions, dropped sends
increase in 25, decrease in three, and stay the same in one, and moving a guard
after its branch has no consistent effect. Phase-end validity never changes,
because the extra sends fall in the tens of cycles between a commit and the
repair.

Send-target checking rejects ten programs. Six defer their update to a timer or
message, which the check is meant to reject. The other four declare
\texttt{Reconfigure} but fail for three reasons.
\texttt{contiki-rplclassic-\linebreak[0]reactive} refreshes its parent on only
one branch of a condition, which also causes its one stale phase.
\texttt{tinyos-ctp-\linebreak[0]reactive} keeps its previous parent while that
parent remains a neighbor. \texttt{zephyr-nud-\linebreak[0]reactive} and
\texttt{zephyr-friend-\linebreak[0]reactive} have a \texttt{Sample} branch that
may reset the destination to the literal \texttt{0}, which never happens in our
schedules. Four of the rejected programs and none of the nine accepted ones end
a phase at a non-neighbor. Taking as positive the destinations that the authors labeled as able to go
stale before running the checker, the decisions have precision 0.70 and recall
1.00. The three false positives are the ctp and the two zephyr programs. MorphC also warns about
the \texttt{Receive} branches of \texttt{meshtastic-nexthop-safe} and
\texttt{openthread-leader}, which overwrite the destination with a
\(\mathsf{dynamic}\) value.

\paragraph{Result on a reduced reproduction of a released defect.}
\label{sec:ordering-reduction}
Contiki-NG issue \#877, one of the four ordering reports of
Section~\ref{sec:field-evidence}, shows a missing check in released code. When
RPL switches parent, a deferred callback sets the TSCH time source to the new
parent, but a keepalive timer that sends to the time source runs first on the
same tick. We reduced the two handlers to C and to MorphLang, with the parent
switch in \texttt{Reconfigure} and the tick and keepalive periods scaled to the
simulator. Over 384 parent changes (32 seeds, 300 switches,
and 84 losses), the C version sends a keepalive to the departed parent after
the node has been notified in 60 changes, and the MorphLang version in none.
Swapping the two C handlers also fixes it, but only as a convention at that
use site.

\paragraph{Answer.}
MorphC places the same one check per branch that a careful C programmer would
write, so its contribution is coverage, not a cheaper check. Omitting the first
guard of each program increases sends to departed relays in all 13; the reduced
Contiki-NG reproduction shows the same failure, and send-target checking
rejects every program whose destination goes stale in our schedules, with three
conservative rejections.

\subsection{RQ3 (Execution and Code-Size Cost)}
\label{sec:rq3}

\begin{wrapfigure}{R}{.395\linewidth}
\includegraphics{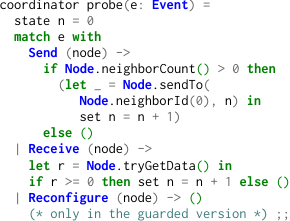}
\caption{The guard-cost microbenchmark. The unguarded version omits the last
branch, so MorphC emits no guards for it.}
\label{fig:guard-bench}
\Description{A MorphLang coordinator named probe with a counter state n. Its
Send branch sends n to the first neighbor and increments it; its Receive branch
polls incoming data and increments n, and an empty Reconfigure branch appears
only in the guarded version.}
\end{wrapfigure}

RQ3 asks about the cost of the guards and the generated code.

\paragraph{Method.}
For the guards alone, we use the microbenchmark in
Figure~\ref{fig:guard-bench}, which we wrote for this purpose and which is not
one of the 19 programs. With one neighbor, no incoming messages, and no topology
change after a 2{,}000-cycle warm-up, we time 1{,}000 loops of both versions.
This measures the path where the revision is unchanged,
not the repair itself. We time the same loop in C (\texttt{gcc -O2}) unguarded,
with an inline check before each branch, and with a runtime dispatcher in a
separate translation unit that checks the revision and calls each handler
through a function pointer, the library alternative of
Section~\ref{sec:design-requirements}.
For whole programs, we use nine of the 19 programs that have a C version,
Programs~1--5, 7, 10, 11, and~15 of Table~\ref{tab:oss-units}, and compare the
two versions when both meet the matching criterion of Section~\ref{sec:setup}. We measure cycles
per issued link transmission, counting sends to departed neighbors and
excluding \texttt{sendTo} calls that find no neighbor, over the whole reference
schedule including the initial phase. Code size counts executable code and
initialized data, not \texttt{.bss} or linker padding. The C programs are
compiled with \texttt{riscv64-elf-gcc}~16.2.0 at \texttt{-O2}, so the ratios
depend on that compiler as well as on the source.

\paragraph{Result.}
The guarded loop takes 88 cycles against 76, so two guards cost 12 cycles,
15.8\% of the unguarded loop, and 144 bytes. In C, an inline check before each
branch adds 14 cycles to a 44-cycle loop, and the runtime dispatcher adds 45.
Only these differences are comparable across compilers, and link-time
inlining could narrow the dispatcher's gap.
Seven of the nine C versions meet the matching criterion as extracted
(Figure~\ref{fig:rq3-cost}, which also includes the OONF pair below).
In all seven, MorphLang takes 1.34 to 2.07
times C's cycles per transmission. Beyond the guards, the difference is consistent
with MorphC's simple code generation, which keeps expression results in
\texttt{a0}, reloads scalar state on each access, and has no peephole pass.
Ordered repair requires none of these choices. We did not isolate these
costs. MorphLang binaries are 0.48 to
0.66 times the size of the five upstream-C binaries, which include data
structures and control flow that the programs omit, and 1.13 to 1.40 times the
three manual translations, whose scope is closer. The \texttt{zephyr-friend}
pair is excluded because its two versions update the destination in different
phases.

\input{figures/fig-rq3-cost}

\paragraph{Result for the OONF pair.}
The C version of \texttt{oonf-olsrv2} fails the matching criterion under
upstream's 1{,}000\,ms rate limiter, which is longer than the whole run. Our timer
shim turns off the limiter without changing the upstream code, and both versions
drop the recomputation on message receipt, whose topology database they omit.
Recomputing once per topology change, MorphLang then takes 0.85 times C's
cycles, and C's count includes upstream routing bookkeeping that the MorphLang
version omits.

\paragraph{Answer.}
On the unchanged-revision path, the guards cost about as much as an inline C
check and less than a runtime dispatcher, and the same 12 cycles separate the
two loops on an FPGA (Section~\ref{sec:fpga}). Seven of the eight matched
programs take 1.34 to 2.07 times C's cycles, which is consistent with code
generation that the contract does not
require, though we did not isolate these costs. Code is 0.48 to 0.66 times the
size of upstream C and 1.13 to 1.40 times that of the manual translations.

\subsection{RQ4 (Results on an FPGA)}
\label{sec:fpga}

RQ4 asks whether the ordering, timing, and outcomes seen in simulation also
hold on an FPGA, where host-injected commits are not synchronized with the core.

\begin{wrapfigure}{R}{.375\linewidth}
\includegraphics{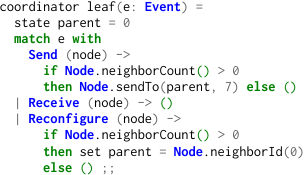}
\caption{The leaf coordinator used in RQ4.}
\label{fig:leaf}
\Description{A MorphLang coordinator named leaf with state parent. Its Send
branch sends the value seven to parent when it has a neighbor; its Receive
branch does nothing, and its Reconfigure branch sets parent to the first
neighbor when there is one.}
\end{wrapfigure}

\paragraph{Method.}
The board holds a single module. We run the RTL of Section~\ref{sec:setup} on a
Digilent Basys~3 board (Xilinx xc7a35t) at 30\,MHz.
There is no second chiplet and no radio. A host PC plays the role of the
relays. Over the board's UART it writes the chiplet's neighbor table, for example
replacing neighbor~5 with neighbor~9, and then commits the change, which
advances the revision as a real topology change would. In the runs that
measure repair windows, the core only sends and receives no messages, and
an on-chip stamp unit records, in core cycles, each commit and the next eight
sends. In the runs of the 19 programs, every neighbor in the committed table
sends periodic beacons generated on the chip, which are the only messages the
core receives, and the board reports a summary of each phase. Because the host is not
synchronized with the core, the UART delay decides where in the loop a commit
lands, and it is not part of any measured interval.
We run three workloads, each compiled by
MorphC and unchanged from simulation (Appendix~\ref{app:basys3}). The first is
the leaf coordinator of Figure~\ref{fig:leaf}, which we wrote for this purpose
and which is not one of the 19 programs. Its \texttt{Reconfigure} branch
re-selects the parent, so the time from a commit to its first send to the new
parent is the repair window. Each of 20 runs attaches one neighbor and then
swaps it 63 times between two cores. The
second is the guarded and unguarded microbenchmark of
Figure~\ref{fig:guard-bench}. The third is each of the 19 programs under the RQ1
reference schedule, which the host replays. These runs therefore test the RTL of one module under
unsynchronized commits, not communication between chiplets or radio behavior,
and not MorphSim's timing.

\paragraph{Result.}
All 9{,}231 sends that the leaf issues under a new revision go to the neighbor
that the commit installed. After 1{,}009 of the 1{,}280 commits, the branch that
was already running issues one more send under the old revision, which the
contract permits because that branch has already passed its guard. Of these,
978 go to the departed neighbor and are dropped, and
31 leave in the commit's own cycle. Every repair window, from a commit to the
first send of the new revision, equals the value that RTL simulation gives for
the loop position where the commit lands, 72 to 114 cycles over all 43
positions, and the loop period stays 43 cycles (Figure~\ref{fig:basys3-windows}).
Here the loop position \(p\) is the number of cycles from the previous send to
the commit. The window falls by one cycle for each cycle that \(p\) grows,
because the commit lands that much closer to the next send. It jumps by one
loop period, 43 cycles, between \(p=10\) and \(p=11\), because that is where
the guard reads the revision. A commit that lands before the read is repaired
in the same iteration, giving \(82-p\) cycles, and one that lands after it
waits for the guard of the next iteration, giving \(125-p=82-p+43\) cycles. The
guarded and unguarded microbenchmarks keep loop periods of 88 and 76 cycles in
all 2{,}880 measured topology changes, the same 12-cycle difference as in
MorphSim.

\input{figures/fig-basys3-windows}

\paragraph{Result for the evaluation programs.}
In all 114 phases of the 19 programs, the board agrees with RTL simulation of
the same script on whether anything was sent, the last destination, whether it
was dropped, and whether anything was delivered. RTL simulation in turn agrees
with MorphSim in 110 of the 114 phases. In the other four, the poll timer of
\texttt{contiki-rplclassic} and \texttt{tinyos-ctp-poll} fires one phase earlier
on the RTL, which moves their single non-neighbor phase without changing its
count. The ablated variants were not run on the board.

\paragraph{Answer.}
On one chiplet's RTL, the ordering, repair windows, guard cost, and
reference-schedule outcomes seen in simulation also hold under host-injected
commits that are not synchronized with the core.

\subsection{Threats to Validity}
\label{sec:threats}

\emph{Construct validity.} Our measures capture whether a destination is
current and how much work is done, not whether a protocol behaves correctly as
a whole. RQ1 scores phase-end destinations, which a later branch can hide, so
RQ2 adds the sends between a commit and its repair. Removing the repair branch
also removes the guards and changes timing, which the body-emptied ablation
controls only partly. Cycles per link transmission measure throughput, not the
latency of a branch, and the guard microbenchmark
(Figure~\ref{fig:guard-bench}) measures only the path on which the revision
is unchanged. The guard count describes MorphC's guards, not programmer effort
or a lower bound for other runtimes (Section~\ref{sec:design-requirements}).

\emph{Internal validity.} The results depend on the simulator, the schedules,
and the C compiler. Agreement with a packet-level reference, RTL simulation,
and the FPGA run reduces but does not remove shared errors in the simulator and
drivers. MorphSim makes link changes and received words visible immediately.
On 12 of the 60 schedules per variant, we delay revision exposure by 90 to
9{,}000 cycles or word delivery by 1 to 100 cycles, one factor at a time. This does
not change which programs improve without repair and shifts the interval of
Section~\ref{sec:dispatch-validation} by at most one branch length. The
receive queue holds 16 words, as in the chiplet's hardware. With an unbounded
queue, words from a departed relay back up across a change. The three
programs that learn their destination from received messages then end two, three, and four
phases at a non-neighbor as written, against three each without repair. The
other ten programs keep their counts. The schedules vary several factors
together, so they show directions rather than the effect of each factor. The
cost ratios of RQ3 also depend on \texttt{riscv64-elf-gcc}~16.2.0.

\emph{External validity.} The programs are our own re-implementations of
selected routines, so the results show what the contract does for these
patterns, not how often they occur or how MorphC would fare on complete
applications. The authors selected the routines and what to omit, wrote the
programs, the two further applications, and the send-target labels, and had no
second annotator. Moving a scheduled callback into \texttt{Reconfigure} runs the
update earlier, so RQ1 evaluates the selected handler, not the upstream
scheduler's delay. Three C baselines are manual translations, and MorphC is not
applied to the original C, C++, or nesC code. The measurements cover cached
numeric destinations and one aggregator, whose repair needed an added
check that accepts readings only from a current neighbor. They do not cover timers or
controller state. The \#877 reproduction scales
the tick and keepalive periods to the simulator, so its rate of 60 in 384 is
not upstream's, and the field reports show that the problem exists, not how
often.

\emph{Scope of the guarantees.} The proofs establish the order of repair, not
that a repair is correct or that MorphC refines the models
(Section~\ref{sec:validation}). Send-target checking covers only numeric
destinations and revisions delivered after a role starts, gives no adjacency guarantee
for runtime-obtained targets, and cannot see state such as OpenThread's global
partition identifier that lies outside local revisions.

\section{Related Work}
\label{sec:related-work}

Table~\ref{tab:related-comparison} compares the mechanisms below by
who writes the update, what orders it before later communication, and
whether the mechanism records revisions and checks numeric destinations.

\input{tables/tab-related-comparison}

\emph{Event-driven languages.} nesC structures asynchronous tasks and events,
where events may preempt tasks and other events, and P checks event-driven state
machines whose programmer can defer events~\cite{gay2003nesc,desai2013p}.
Neither ties an application callback to a per-role topology revision with a
compiler check before each later communication branch.

\emph{Protocol and aggregate languages.} Choral, Pirouette, HasChor, dynamic
session types, and process calculi check interactions as participants or
protocols change
\cite{giallorenzo2024choral,hirsch2022pirouette,shen2023haschor,denielou2011dynamic,
harvey2021ensembles,dallapreda2017dynamic,lopez2016available,kouzapas2011dynamic}.
Choreographic theories and polymorphic choreographies generalize who
participates \cite{cruzfilipe2023theory,graversen2024polychor}. Choreographic
languages guarantee communication safety and deadlock freedom, but their proofs
assume reliable communication~\cite{giallorenzo2024choral}, which a permitted
send to a neighbor retained from an earlier graph does not satisfy. Proto, field calculi,
and Protelis define
computations that converge as neighborhoods change
\cite{viroli2013proto,viroli2014fields,viroli2019aggregate,pianini2015protelis},
and macroprogramming systems such as Regiment program a sensor network as a
whole while the runtime maintains neighborhoods
\cite{newton2007regiment,mottola2011programming}. Both hide the per-node
update that MorphLang exposes.

\emph{Languages for modular robots.} Shape-changing ensembles have been
programmed declaratively. Meld derives facts from rules over facts that
include each module's current neighbors, which the runtime maintains, and LDP
expresses programs as predicates over small connected groups of modules that
are matched as the ensemble changes shape
\cite{ashleyrollman2007meld,derosa2008ldp}. Both recompute derived state from
the current neighbors. MorphLang instead keeps imperative persistent state,
as existing network stacks do, and orders the application's own repair of it.

\emph{Membership and adaptation.} Virtual synchrony gives all members of a
process group the same order of message deliveries and membership changes, and
ULFM reports failures per operation and lets the surviving processes revoke,
shrink, and agree on a communicator
\cite{birman1987virtual,birman1993process,bland2013ulfm,mpi41}.
MorphLang's contract
instead orders repair within one role, without a global barrier or a common
view. Erlang supervision, AmbientTalk, ContextErlang, EventCJ, Concerto, and
related systems make failure handling or adaptation explicit but leave the
state update and its timing to the application
\cite{kephart2003autonomic,erlangOtpSupervision,vancutsem2014ambienttalk,
salvaneschi2015contexterlang,kamina2011eventcj,chardet2021concerto,
arfi2024concertod,ahrens2022reconfig}.

\emph{Reactive programming.} FRP and dataflow languages recompute derived
values when their inputs change, and glitch freedom ensures that no
computation observes a mix of values derived from old and new inputs
\cite{cooper2006embedding,bainomugisha2013survey}. The ordered-repair
contract is a local analog for topology. Because the dependent state cannot
be recomputed from a dataflow graph, the application writes the repair and
MorphC orders it.

\emph{Dynamic software updating.} DSU systems replace a running program's code
and transform its state with programmer-written transformers at update points.
Version consistency restricts when an update may take effect
\cite{hicks2005dsu,neamtiu2006ginseng,hayden2012kitsune,neamtiu2008contextual}.
\texttt{Reconfigure} resembles a state transformer, and the guards resemble
update points, but a topology change leaves the code unchanged and recurs, so
MorphC coalesces pending changes per role. Read-copy update and epoch-based
reclamation also track when readers stop using an old version
\cite{mckenney1998rcu,fraser2004epoch}.

\emph{Flow analyses.} Origin inference tracks every source a value may come from, in
the style of information-flow type systems
\cite{volpano1996sound,sabelfeld2003ifc}, with \(\mathsf{fixed}\) as the
label that rejects a send, and definite refresh must hold on every path, like
definite assignment. Unlike a security label, the origin of a persistent
location depends on the delivered revision.

\emph{Simulators.} TOSSIM simulates TinyOS applications at the level of events
and packets~\cite{levis2003tossim}, Cooja adds emulation of node
hardware~\cite{osterlind2006cooja}, Avrora runs AVR binaries cycle by
cycle~\cite{titzer2005avrora}, and gem5 models processors but not a changing
network~\cite{binkert2011gem5}. MorphSim runs RV32I binaries cycle by cycle
like Avrora, and its topology changes are deterministic events that advance a
per-core revision at any chosen cycle.

\section{Conclusion}
\label{sec:conclusion}

We introduced morph-aware computing, a programming model for systems whose
computing modules keep running and keep their state while the topology changes.
The model divides the work. Applications state what to repair in a
\texttt{Reconfigure} branch of MorphLang, and the compiler decides when the
repair runs.

MorphC compiles MorphLang into bare-metal RISC-V RV32I code. It guards every
ordinary branch so that a role runs its repair for the latest revision it has
seen before any other branch, and it rejects numeric send targets that may be
literal or unrefreshed. For a model of the generated code, we proved in Rocq
that this ordered-repair contract holds for any number of roles and topology
changes, that each delivered revision is repaired at most once, and that
send-target checking is sound for a core language.

In evaluations of 19 programs derived from seven network stacks, removing the
repair reduced destination validity in ten of the 13 programs that declare one.
The guards cost 12 cycles per loop in a microbenchmark. A reduced reproduction
of a released Contiki-NG defect showed the failure that the guards prevent, and
an FPGA implementation of one module agreed with RTL simulation.

The order of updates can be part of a language's semantics rather than a
convention of its runtime. Coverage then follows from program structure, since
every ordinary branch is visible to the compiler, and because repairs run
first, a static analysis can rely on their outcomes.

We see three directions for future work. First, we aim to broaden MorphC's
checks to other stale state, such as table entries for departed neighbors and
messages queued before a change (Section~\ref{sec:aggregator}). Letting
programmers mark such tables and queues would let MorphC require the repair to
clear them. Committing an initial revision before the cores start would extend
the contract to the initial state. Next, we intend to connect our proofs to the
compiler by proving in Rocq that MorphC's guard insertion always yields the
structure that the checker now validates for each binary. Finally, we plan to
evaluate our approach on many distributed modules that connect and disconnect
over wireless links, rather than on a single module driven by a host. This also
requires a method for deploying compiled role binaries to many modules, which
MorphC does not yet provide.

\section*{Data-Availability Statement}

The artifact contains the implementation, fixed upstream commits,
evaluation inputs and scripts, archived field reports, formal developments,
RTL/testbench sources, and the bitstream flow, host scripts, and raw stamps of
the FPGA run (Section~\ref{sec:fpga}).
\ifarxiv The artifact is not yet publicly available. We plan to release it
with an archival DOI and will update this entry when it is available.\else We
intend to submit the artifact for Artifact Evaluation.\fi

\section*{Use of Generative AI Tools}

We used Claude, a large language model by Anthropic, to proofread and polish the
text of this paper and to assist in implementing parts of the artifact's code
and scripts. The authors reviewed and edited all generated content and take full
responsibility for the paper and the artifact.

\bibliographystyle{ACM-Reference-Format}
\bibliography{reference}

\ifdefined\splitappendix\else
  \appendix
\input{appendix}
\fi

\end{document}

%% file: preamble.tex
\usepackage{enumitem}
\usepackage{booktabs}
\usepackage{tabularx}
\usepackage{pifont}
\usepackage{hyperref}
\usepackage{xurl}      
\usepackage{fvextra}
\newif\ifarxiv\arxivtrue
\usepackage{xcolor}
\usepackage{wrapfig}
\usepackage{placeins}
\usepackage{tikz}
\usepackage{pgfplots}
\pgfplotsset{compat=1.18}
\usetikzlibrary{arrows.meta,positioning,calc,fit,backgrounds,decorations.pathmorphing}

\colorlet{figink}{black!82}
\colorlet{figrule}{black!30}
\colorlet{figmute}{black!64}
\colorlet{figaccent}{blue!64!black}
\colorlet{figok}{green!45!black}
\colorlet{figbad}{red!62!black}
\definecolor{figrq4sim}{HTML}{2A78D6}
\definecolor{figrq4fpga}{HTML}{EB6834}
\definecolor{figrq3cyc}{HTML}{2A78D6}
\definecolor{figrq3size}{HTML}{EB6834}
\definecolor{pipefill}{HTML}{F3F1EA}
\definecolor{pipeend}{HTML}{E7E4DB}
\definecolor{chiptop}{HTML}{F4F2EC}
\definecolor{chipside}{HTML}{E3E0D8}
\tikzset{
  sedge/.style={draw=figmute, line width=0.5pt},
  shair/.style={draw=figrule, line width=0.4pt},
  scoil/.style={draw=figrule, line width=0.3pt},
  tw/.style={anchor=base west, font=\tiny, inner sep=0pt, outer sep=0pt},
  tc/.style={anchor=base, font=\tiny, inner sep=0pt, outer sep=0pt},
  sw/.style={anchor=base west, font=\scriptsize, inner sep=0pt, outer sep=0pt},
  swave/.style={-{Latex[length=1.3mm]}, figaccent, line width=0.6pt, decorate,
    decoration={snake, amplitude=0.4mm, segment length=1.4mm,
                pre length=0.3mm, post length=0.9mm}},
  swavep/.style={figaccent, line width=0.6pt, decorate,
    decoration={snake, amplitude=0.4mm, segment length=1.4mm,
                pre length=0.3mm, post length=0.3mm}},
  sghost/.style={draw=black!45, line width=0.45pt,
                 dash pattern=on 1.8pt off 1.3pt},
  stepnum/.style={draw=figink, circle, inner sep=0.9pt,
                  font=\tiny\bfseries, text=figink},
  codecallout/.style={draw=figaccent, fill=figaccent!6, rounded corners=1pt,
                      inner xsep=2pt, inner ysep=0.7pt,
                      font=\sffamily\tiny\bfseries, text=figaccent},
}

\tikzset{
  ftitle/.style={font=\scriptsize\bfseries, text=figink, anchor=base west,
                 inner sep=0pt},
  fhair/.style={draw=figrule, line width=0.4pt},
  fbody/.style={font=\scriptsize, text=figink, anchor=north west, align=left,
                inner sep=0pt},
  flabel/.style={font=\tiny, text=figmute, anchor=base, inner sep=0pt,
                 align=center},
  fnote/.style={font=\tiny, text=figmute, anchor=north west, align=left,
                inner sep=0pt},
  fkey/.style={font=\scriptsize, text=figmute, anchor=base west, inner sep=0pt},
  farrow/.style={-{Latex[length=1.4mm]}, line width=0.5pt, figmute},
  die/.style={draw=figink, line width=0.5pt},
  hotdie/.style={draw=figaccent, line width=0.9pt},
  gonedie/.style={draw=figrule, line width=0.4pt, dashed},
  coil/.style={draw=figrule, line width=0.3pt},
  logic/.style={draw=figmute, line width=0.3pt, inner sep=0pt,
                font=\scriptsize, align=center, text=figink},
  link/.style={-{Latex[length=1.3mm]}, line width=0.6pt, figaccent, decorate,
               decoration={snake, amplitude=0.28mm, segment length=1.05mm,
                           pre length=0.35mm, post length=0.85mm}},
}

\DeclareRobustCommand{\checkmark}{\ding{51}}

\ifdefined\splitappendix
  \usepackage{xr-hyper}
  \makeatletter
  \@for\@tempa:=-1,0,1,2,3,4,5\do{%
    \global\expandafter\let\csname r@tocindent\@tempa\endcsname\relax}
  \global\expandafter\let\csname r@TotPages\endcsname\relax
  \makeatother
\fi

\newcommand{\papertitle}{%
  Morph-Aware Computing: Compiler-Ordered Updates of Program State after Topology Changes}

\author{Yusuke Izawa}
\affiliation{
  \institution{Tokyo Metropolitan University}
  \city{Tokyo}
  \country{Japan}
}
\affiliation{
  \institution{Heinrich-Heine-Universit\"at D\"usseldorf}
  \city{D\"usseldorf}
  \country{Germany}
}
\author{Junichiro Kadomoto}
\affiliation{
  \institution{The University of Tokyo}
  \city{Tokyo}
  \country{Japan}
}

%% file: figures/fig-teaser-float.tex
\begin{figure}[t]
\centering
\includegraphics[width=0.75\columnwidth]{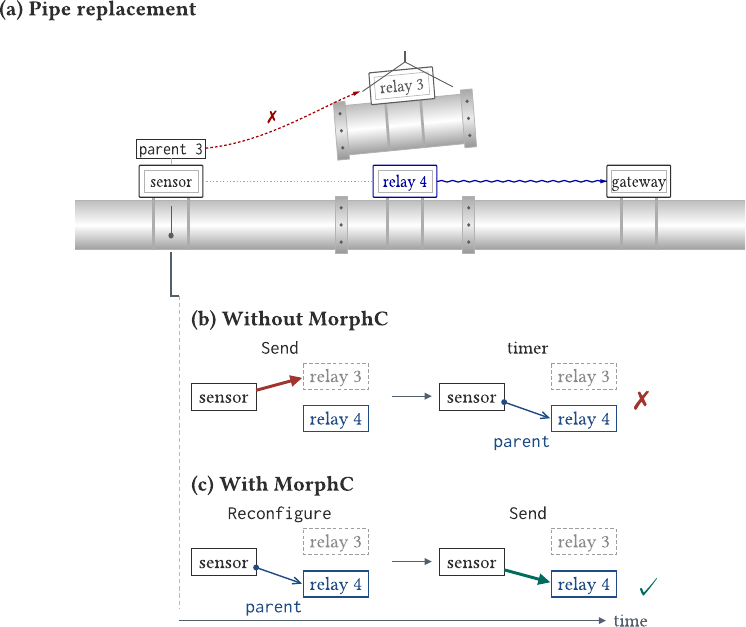}
\caption{An illustrative pipe-maintenance scenario (not to scale). (a) Workers
replace a spool. Relay~3 leaves, relay~4 arrives, and the still-running sensor
keeps \texttt{parent=3} although its only neighbor is relay~4. (b, c) The same
sensor after the swap, programmed without and with MorphC, on one time axis.
Thin blue arrows show the stored \texttt{parent}, thick arrows a send. Without
MorphC a send can precede the timer-triggered update of \texttt{parent}, the
trigger that 4 of the 13 upstream routines use. With MorphC,
\texttt{Reconfigure} runs before \texttt{Send}.}
\label{fig:teaser}
\Description{A corrosion sensor, relay four, and a gateway remain clamped to
a pipe. Above them, a removed spool hangs from a lifting sling, still carrying
relay three. A dashed red arrow from the sensor's retained parent-three
register points to that departed relay. The sensor's new one-hop neighbor is
relay four, and the count remains one. Below, two rows start from the sensor after the swap and share a time axis.
Each step is drawn as a small copy of the scene, with the sensor, departed relay
three in dashed gray, and relay four. In row b, without MorphC, a thick red send arrow goes to
departed relay three, and only afterwards a timer-driven update moves the thin
parent arrow to relay four. In row c, with MorphC, Reconfigure first moves the parent arrow
to relay four, and the thick green send arrow then reaches relay four. This is an illustrative use case, not a reported deployment.}
\end{figure}

%% file: figures/fig-stack-compact.tex
\begin{figure}[t]
\centering
\def\morphstackcompact{}
\includegraphics{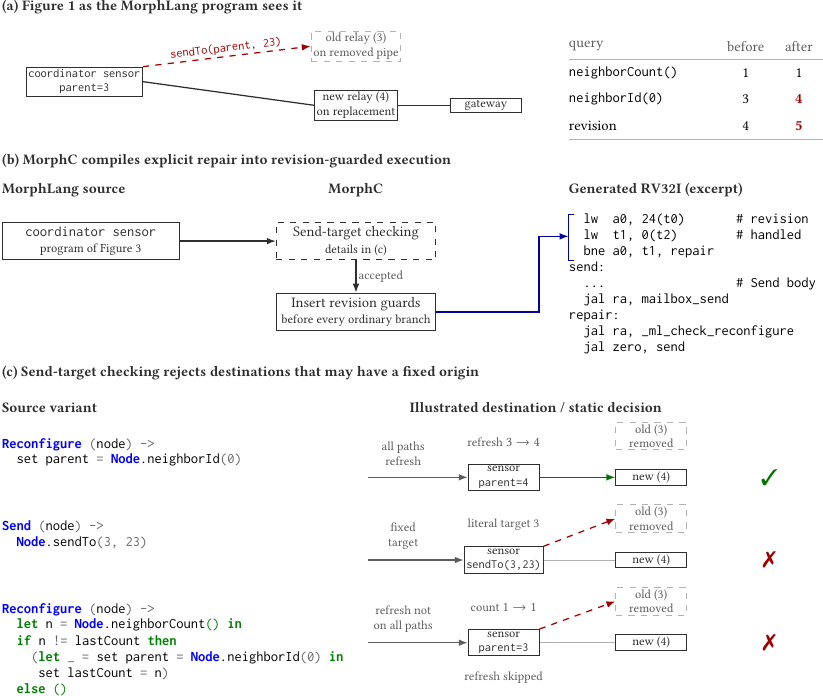}
\caption{Figure~\ref{fig:teaser}'s sensor through MorphC. (a) The replacement
as the program observes it. The neighbor count is unchanged, but the stored
\texttt{parent} names the departed relay. (b) MorphC checks send
targets, then inserts a revision guard before every ordinary branch. On a
mismatch, the helper runs the repair and records completion before returning to
\texttt{send}. (c) The
check on three variants of (b). It inspects source-level origins, not the
pictured topology. Arrows are send targets, and acceptance (\ding{51}) does
not prove adjacency.}
\label{fig:software-stack}
\Description{Three panels show MorphC applied to the sensor of Figure 1.
Panel a shows the replacement as a network in MorphLang terms, namely the
coordinator sensor whose parent is three, a dashed send arrow
labeled sendTo parent to the old relay three on the removed pipe, and a link
to the new relay four and the gateway. A note states that
Node.neighborCount stays one, Node.neighborId of zero changes from three to
four, and revision four to five makes Reconfigure pending.
Panel b shows a compilation flow. The MorphLang sensor program of Figure 3 goes first
through static send-target checking (detailed in panel c). An
accepted program then goes to revision-guard insertion, which leads to
generated RV32I. The highlighted loads and conditional branch detect an
unhandled revision. The repair helper reads the revision and snapshot again, runs application
repair, and records completion before returning to Send. Panel c shows three
source variants, each paired with a destination sketch using the same
positions for the sensor and old and new relays. Unconditional refresh
changes parent from three to four and points to the new relay. It is
accepted. A literal destination of three points to the removed relay and is
rejected. Count-based refresh keeps parent three because the count remains
one. Its arrow also points to the removed relay and it is rejected. A check
mark denotes acceptance and a cross denotes rejection. Plain lines show
neighbor links, and arrows show send targets.}
\end{figure}

%% file: tables/tab-programs.tex
\begin{table}[t]
\caption{The 13 upstream functions or handlers and their MorphLang programs.
Numbers index all 19 programs. The lower part lists the six derived variants
and the program each is written from.
\emph{Rec.} marks programs that declare \texttt{Reconfigure} as written.
\emph{Der.} gives the derived variant, which updates on a topology revision and is
named with the suffix \texttt{-reactive} (\texttt{-safe} for program~14).
\emph{C} gives the RQ3 baseline, upstream C unchanged (U) or a manual translation (M).}
\label{tab:oss-units}
\scriptsize
\begin{tabularx}{\columnwidth}{@{}rl>{\raggedright\arraybackslash}X>{\raggedright\arraybackslash}p{0.235\columnwidth}ccc@{}}
\toprule
\# & Program & Upstream function or handler & Upstream update trigger & Rec. & Der. & C \\
\midrule
1 & \texttt{contiki-rpllite} & Contiki-NG \texttt{rpl\_neighbor\_select\_best} & neighbor-removal callback & \checkmark & -- & U \\
2 & \texttt{contiki-rpllite-repair} & Contiki-NG \texttt{rpl\_local\_repair} & newer-version packet & \checkmark & -- & U \\
3 & \texttt{contiki-trickle-reset} & Contiki-NG \texttt{rpl\_timers\_dio\_reset} & local-repair call & \checkmark & -- & U \\
4 & \texttt{contiki-rplclassic} & Contiki-NG \texttt{rpl\_recalculate\_ranks} & feedback flag and timer & -- & 16 & U \\
5 & \texttt{tinyos-ctp-poll} & TinyOS CTP \texttt{updateRouteTask} & route/beacon timer & -- & 6 & M \\
7 & \texttt{openthread-routertable} & OpenThread \texttt{RouterTable::UpdateRoutes} & route-update packet & \checkmark & -- & M \\
8 & \texttt{openthread-parentsearch} & OpenThread \texttt{ParentSearch::HandleTimer} & parent-search timer & -- & 17 & -- \\
9 & \texttt{openthread-leader} & OpenThread \texttt{HandleAdvertisementOnFtd} & leader advertisement & \checkmark & -- & -- \\
10 & \texttt{zephyr-friend} & Zephyr BT Mesh \texttt{friend\_timeout} & friendship timeout & -- & 19 & U \\
11 & \texttt{zephyr-nud} & Zephyr \texttt{ipv6\_nd\_reachable\_timeout} & reachability timer & -- & 18 & M \\
12 & \texttt{riot-gnrcrpl} & RIOT \texttt{\_gnrc\_rpl\_find\_\allowbreak preferred\_parent} & parent-removal callback & \checkmark & -- & -- \\
13 & \texttt{meshtastic-nexthop} & Meshtastic \texttt{NextHopRouter::getNextHop} & ack/reply packet & -- & 14 & -- \\
15 & \texttt{oonf-olsrv2} & OONF OLSRv2 \texttt{\_run\_dijkstra} & metric/change/expiry callback & \checkmark & -- & U \\
\midrule
\multicolumn{7}{@{}l}{\emph{Derived variants}, each written from the program in the third column} \\
6 & \texttt{tinyos-ctp-reactive} & from 5, update body moved into \texttt{Reconfigure} & topology change & \checkmark & -- & -- \\
16 & \texttt{contiki-rplclassic-reactive} & from 4, update body moved into \texttt{Reconfigure} & topology change & \checkmark & -- & -- \\
17 & \texttt{openthread-parentsearch-reactive} & from 8, scan also called from \texttt{Reconfigure} & timer and topology change & \checkmark & -- & -- \\
14 & \texttt{meshtastic-nexthop-safe} & from 13, \texttt{Reconfigure} added & topology change & \checkmark & -- & -- \\
18 & \texttt{zephyr-nud-reactive} & from 11, \texttt{Reconfigure} added & topology change & \checkmark & -- & -- \\
19 & \texttt{zephyr-friend-reactive} & from 10, \texttt{Reconfigure} added & topology change & \checkmark & -- & -- \\
\bottomrule
\end{tabularx}
\end{table}

%% file: tables/tab-rq1-exact.tex
\begin{table}[t]
\caption{RQ1 under the reference schedule, showing the phases, of six, that end with the cached
destination not in the neighbor snapshot while at least one relay is adjacent. \emph{No repair} removes the
\texttt{Reconfigure} branch, and N/A means that the program has none. Programs 6, 14, and
16--19 are derived variants (Table~\ref{tab:oss-units}).}
\label{tab:rq1-exact}
\scriptsize
\resizebox{\columnwidth}{!}{%
\begin{tabular}{@{}rlrr@{\qquad}rlrr@{}}
\toprule
\# & MorphLang program & As written & No repair &
\# & MorphLang program & As written & No repair \\
\midrule
1 & \texttt{contiki-rpllite}         & 0/6 & 6/6 & 11 & \texttt{zephyr-nud}              & 0/6 & N/A \\
2 & \texttt{contiki-rpllite-repair}  & 0/6 & 6/6 & 12 & \texttt{riot-gnrcrpl}            & 0/6 & 6/6 \\
3 & \texttt{contiki-trickle-reset}   & 0/6 & 6/6 & 13 & \texttt{meshtastic-nexthop}      & 0/6 & N/A \\
4 & \texttt{contiki-rplclassic}      & 1/6 & N/A & 14 & \texttt{meshtastic-nexthop-safe} & 0/6 & 0/6 \\
5 & \texttt{tinyos-ctp-poll}         & 1/6 & N/A & 15 & \texttt{oonf-olsrv2}             & 0/6 & 6/6 \\
6 & \texttt{tinyos-ctp-reactive}     & 0/6 & 6/6 & 16 & \texttt{contiki-rplclassic-\linebreak[0]reactive} & 1/6 & 6/6 \\
7 & \texttt{openthread-routertable}  & 0/6 & 6/6 & 17 & \texttt{openthread-parentsearch-\linebreak[0]reactive} & 0/6 & 3/6 \\
8 & \texttt{openthread-parentsearch} & 3/6 & N/A & 18 & \texttt{zephyr-nud-\linebreak[0]reactive}     & 0/6 & 0/6 \\
9 & \texttt{openthread-leader}       & 0/6 & 4/6 & 19 & \texttt{zephyr-friend-\linebreak[0]reactive}  & 0/6 & 0/6 \\
10 & \texttt{zephyr-friend}          & 0/6 & N/A &    &                                     &     &     \\
\bottomrule
\end{tabular}%
}
\end{table}

%% file: tables/tab-rq1-sequence.tex
\begin{table}[t]
\caption{Direction of the RQ1 ablations over 20 replacement sequences, each at three
phase lengths. W, U, and B mean worse, unchanged, and better at all three lengths, and L means that the direction depends
on length. Each group sums to 20. The first row covers programs 1--3, 6, 7, 12,
15, and 16 individually, the last programs 14, 18, and 19.}
\label{tab:rq1-sequences}
\small
\setlength{\tabcolsep}{4pt}
\begin{tabularx}{\columnwidth}{@{}Xrrrrrrrr@{}}
\toprule
Program & \multicolumn{4}{c}{Branch removed} &
\multicolumn{4}{c}{Body emptied} \\
\cmidrule(lr){2-5}\cmidrule(l){6-9}
 & W & U & B & L & W & U & B & L \\
\midrule
Eight always-worse programs (each) & 20 & 0 & 0 & 0 & 20 & 0 & 0 & 0 \\
\texttt{openthread-leader} & 12 & 8 & 0 & 0 & 12 & 8 & 0 & 0 \\
\texttt{openthread-parentsearch-reactive} & 15 & 0 & 0 & 5 & 15 & 0 & 0 & 5 \\
Three programs that learn from received messages (each) & 0 & 20 & 0 & 0 & 0 & 20 & 0 & 0 \\
\bottomrule
\end{tabularx}
\end{table}

%% file: figures/fig-rq3-cost.tex
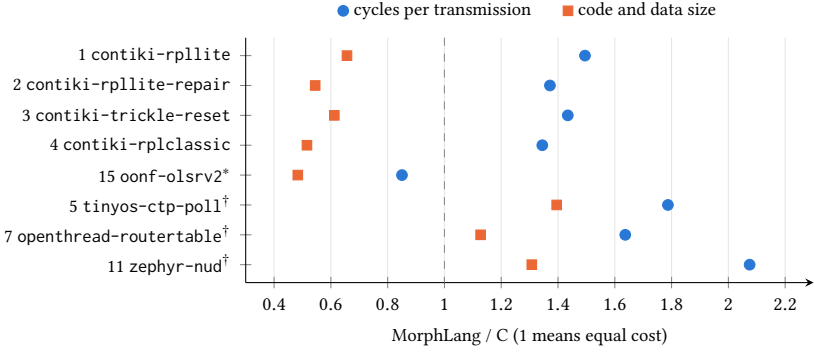
\begin{figure}[t]
\centering
\input{figures/fig-rq3-cost.tikz}
\caption{Cost of the eight MorphLang and C pairs that meet the matching
criterion, as MorphLang/C ratios on the reference schedule. Programs 1 to 4
use upstream C unchanged, Program 15 (\(^{*}\)) uses upstream C with its rate
limiter disabled, and \(^{\dagger}\) marks a manual C translation. Appendix
Table~\ref{tab:rq3-oss-cost} gives the values.}
\label{fig:rq3-cost}
\Description{A dot plot with one row per program and a dashed line at ratio 1.
Cycle ratios lie between 1.34 and 2.07 for seven programs and at 0.85 for
oonf-olsrv2. Size ratios lie between 0.48 and 0.66 for the five upstream-C
programs and between 1.13 and 1.40 for the three manual translations.}
\end{figure}

%% file: figures/fig-rq3-cost.tikz
\begin{tikzpicture}
\begin{axis}[
  scale=0.9,
  scale only axis, width=0.6\columnwidth, height=3.6cm,
  xmin=0.3, xmax=2.3, ymin=0.4, ymax=8.6,
  xlabel={MorphLang / C (1 means equal cost)},
  ytick={1,...,8}, y dir=reverse,
  yticklabels={
    {1 \texttt{contiki-rpllite}},
    {2 \texttt{contiki-rpllite-repair}},
    {3 \texttt{contiki-trickle-reset}},
    {4 \texttt{contiki-rplclassic}},
    {15 \texttt{oonf-olsrv2}\(^{*}\)},
    {5 \texttt{tinyos-ctp-poll}\(^{\dagger}\)},
    {7 \texttt{openthread-routertable}\(^{\dagger}\)},
    {11 \texttt{zephyr-nud}\(^{\dagger}\)}},
  tick label style={font=\scriptsize}, label style={font=\scriptsize},
  axis lines=left, y axis line style={-}, xmajorgrids, grid style={draw=black!10},
  legend style={font=\scriptsize, draw=none, fill=none, at={(0.5,1.03)}, anchor=south}, legend columns=2, legend style={/tikz/every even column/.append style={column sep=1em}},
  legend cell align=left,
]
\draw[black!45, dashed] (axis cs:1,0.4) -- (axis cs:1,8.6);
\addplot[only marks, mark=*, mark size=2pt, color=figrq3cyc]
  table[x=morphl_over_c_cycles, y expr=\coordindex+1, col sep=comma, row sep=crcr] {program_number,program,c_kind,morphl_cycles_per_tx,c_cycles_per_tx,morphl_over_c_cycles,morphl_bytes,c_bytes,morphl_over_c_bytes\\1,contiki-rpllite,upstream,98.09,65.61,1.4949400591623851,1504,2290,0.6567685589519651\\2,contiki-rpllite-repair,upstream,85.24,62.15,1.3715329454742253,1388,2548,0.5447409733124019\\3,contiki-trickle-reset,upstream,115.64,80.61,1.4344713656387664,1780,2908,0.6121045392022009\\4,contiki-rplclassic,upstream,203.88,151.62,1.3446601941747574,1608,3118,0.5157152020525978\\15,oonf-olsrv2,upstream-ratelimit-off,116.77,137.31,0.8504170528266913,1524,3152,0.48350253807106597\\5,tinyos-ctp-poll,manual,190.22,106.44,1.7871376811594204,1920,1376,1.3953488372093024\\7,openthread-routertable,manual,105.09,64.2,1.636864053377815,1592,1412,1.1274787535410764\\11,zephyr-nud,manual,129.36,62.35,2.0749486652977414,1768,1352,1.3076923076923077\\};
\addlegendentry{cycles per transmission}
\addplot[only marks, mark=square*, mark size=1.8pt, color=figrq3size]
  table[x=morphl_over_c_bytes, y expr=\coordindex+1, col sep=comma, row sep=crcr] {program_number,program,c_kind,morphl_cycles_per_tx,c_cycles_per_tx,morphl_over_c_cycles,morphl_bytes,c_bytes,morphl_over_c_bytes\\1,contiki-rpllite,upstream,98.09,65.61,1.4949400591623851,1504,2290,0.6567685589519651\\2,contiki-rpllite-repair,upstream,85.24,62.15,1.3715329454742253,1388,2548,0.5447409733124019\\3,contiki-trickle-reset,upstream,115.64,80.61,1.4344713656387664,1780,2908,0.6121045392022009\\4,contiki-rplclassic,upstream,203.88,151.62,1.3446601941747574,1608,3118,0.5157152020525978\\15,oonf-olsrv2,upstream-ratelimit-off,116.77,137.31,0.8504170528266913,1524,3152,0.48350253807106597\\5,tinyos-ctp-poll,manual,190.22,106.44,1.7871376811594204,1920,1376,1.3953488372093024\\7,openthread-routertable,manual,105.09,64.2,1.636864053377815,1592,1412,1.1274787535410764\\11,zephyr-nud,manual,129.36,62.35,2.0749486652977414,1768,1352,1.3076923076923077\\};
\addlegendentry{code and data size}
\end{axis}
\end{tikzpicture}

%% file: figures/fig-basys3-windows.tex
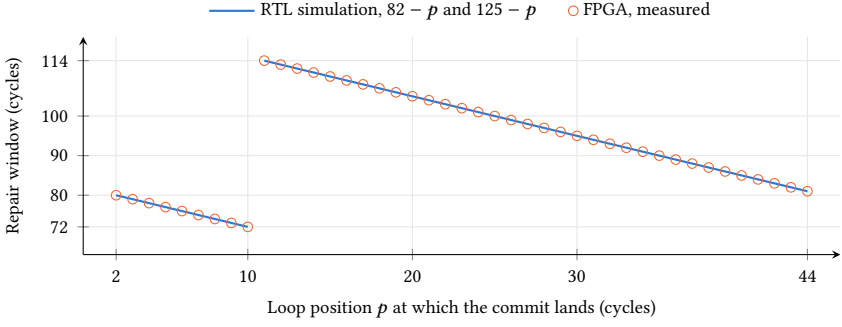
\begin{figure}[t]
\centering
\input{figures/fig-basys3-windows.tikz}
\caption{Repair windows of the leaf coordinator (Figure~\ref{fig:leaf}) running
on one chiplet's RTL on an FPGA, for the 1{,}260 swaps in which a host PC
replaces the node's only neighbor, by the loop position at which the commit lands.
A repair window is the number of cycles from a commit to the first send to the
new neighbor. Each circle stands for the 19 to 39 swaps at one position, and
every measured window lies on the RTL-simulation value, \(82-p\) for a commit
before the guard reads the revision and \(125-p\) for one after it, one loop
period later. The 20 attaches, not
shown, also match simulation, at 72 to 93 cycles.}
\label{fig:basys3-windows}
\Description{A scatter plot of repair window in cycles against loop position
from 2 to 44. A line falls from 80 at position 2 to 72 at position 10, jumps to
114 at position 11, and falls to 81 at position 44. Every measured FPGA point
lies on the line.}
\end{figure}

%% file: figures/fig-basys3-windows.tikz
\begin{tikzpicture}
\begin{axis}[
  scale=0.9,
  scale only axis, width=0.8\columnwidth, height=3.2cm,
  xmin=0, xmax=46, ymin=65, ymax=120,
  xlabel={Loop position \(p\) at which the commit lands (cycles)},
  ylabel={Repair window (cycles)},
  xtick={2,10,20,30,44}, ytick={72,80,90,100,114},
  tick label style={font=\scriptsize}, label style={font=\scriptsize},
  axis lines=left, grid=major, grid style={draw=black!10},
  legend style={font=\scriptsize, draw=none, fill=none, at={(0.5,1.03)}, anchor=south}, legend columns=2, legend style={/tikz/every even column/.append style={column sep=1em}},
  legend cell align=left,
]
\addplot[domain=2:10, samples=2, thick, color=figrq4sim, forget plot] {82-x};
\addplot[domain=11:44, samples=2, thick, color=figrq4sim] {125-x};
\addlegendentry{RTL simulation, \(82-p\) and \(125-p\)}
\addplot[only marks, mark=o, mark size=1.8pt, color=figrq4fpga]
  table[x=p, y=window, col sep=comma, row sep=crcr] {p,window,count\\2,80,21\\3,79,33\\4,78,29\\5,77,24\\6,76,23\\7,75,21\\8,74,29\\9,73,19\\10,72,28\\11,114,26\\12,113,30\\13,112,32\\14,111,31\\15,110,39\\16,109,27\\17,108,37\\18,107,27\\19,106,27\\20,105,23\\21,104,21\\22,103,31\\23,102,32\\24,101,29\\25,100,29\\26,99,33\\27,98,34\\28,97,33\\29,96,25\\30,95,32\\31,94,35\\32,93,37\\33,92,29\\34,91,33\\35,90,32\\36,89,29\\37,88,34\\38,87,36\\39,86,26\\40,85,28\\41,84,29\\42,83,32\\43,82,31\\44,81,24\\};
\addlegendentry{FPGA, measured}
\end{axis}
\end{tikzpicture}

%% file: tables/tab-related-comparison.tex
\begin{table}[t]
\caption{Mechanisms that respond to communication-structure changes, by their
usual guarantee on the four dimensions of this paper.}
\label{tab:related-comparison}
\scriptsize
\begin{tabularx}{\columnwidth}{@{}>{\raggedright\arraybackslash}p{0.22\columnwidth}>{\raggedright\arraybackslash}X>{\raggedright\arraybackslash}X>{\raggedright\arraybackslash}X>{\raggedright\arraybackslash}X@{}}
\toprule
Mechanism & Update & Order & Revision & Target check \\
\midrule
Event callbacks & callback & convention & no & no \\
Choreography/session types & protocol state & protocol actions & no & protocol only \\
Membership/recovery & application code & view/message order & view-level & no \\
Aggregate/field calculi & convergence rule & no local point & implicit & no \\
MorphLang/MorphC & \texttt{Reconfigure} & compiler guard & per role & numeric \\
\bottomrule
\end{tabularx}
\end{table}

%% file: appendix.tex
\section{Supporting Evidence and Technical Details}

Full program listings, line-level source records, mutation details, and
reproduction scripts belong to the artifact of the Data-Availability Statement,
which is not yet publicly available.

\subsection{Claims and Evidence}
\label{app:claim-evidence}

Table~\ref{tab:claim-evidence} maps each claim of Sections~\ref{sec:validation}
and~\ref{sec:evaluation} to its evidence and the source of its inputs.
\begin{table}[h]
\caption{Claims, evidence, and provenance for validation and evaluation.}
\label{tab:claim-evidence}
\scriptsize
\begin{tabularx}{\columnwidth}{@{}p{0.21\columnwidth}Xp{0.20\columnwidth}@{}}
\toprule
Component & Evidence & Source of inputs \\
\midrule
Explicit repair (R1) & RQ1 ablations and attach-once control
(Section~\ref{sec:whole-app}). & Upstream routines \\
Ordered repair (R2) & TLA+, Rocq, binary checker, timing test, stale sends
of guard mutants (RQ2), and the FPGA run (RQ4). & Authors only \\
State continuity (R3) & Theorem~\ref{thm:continuity}
(Section~\ref{sec:dispatch-validation}). & Authors only \\
Coordination (R4) & Generated/manual guard sites (RQ2). & Authors only \\
Send-target checking & Rocq proofs
(Section~\ref{sec:send-target-validation}) and RQ2. & Upstream routines and authors' labels \\
Implementation cost & Microbenchmark and C pairs (RQ3). & Upstream C and authors' C \\
Simulator final values & 48 kernel/pipeline comparisons
(Section~\ref{sec:lower-layers}). & QEMU and native execution \\
FPGA ordering and timing & 1{,}280 host-injected commits on one node and
19 programs, 114 phases (RQ4). & Authors only \\
Stale-state failures & 13 field reports, classified manually (Section~\ref{sec:field-evidence}). & Third-party reports \\
\bottomrule
\end{tabularx}
\end{table}




\subsection{Mapping from Upstream Code to MorphLang}
\label{app:oss-translation-map}

Table~\ref{tab:oss-translation-map} records what the authors retained and
changed in each of the 13 upstream programs. We wrote the programs manually from
upstream source, not from the C timing executables. For each upstream routine,
we identify the cached destination, its reads and updates, and the event or
callback that reaches them. We map that event to \texttt{Send},
\texttt{Receive}, \texttt{Sample}, or \texttt{Reconfigure} and keep the
selected assignments and conditionals in order. Each program re-implements
selected behavior rather than translating line by line, and its file documents
the replaced dependencies, such as ETX calculations, framework data structures,
and delaying callbacks. The table does not establish automatic translation or
whole-application equivalence.

\begin{table}[h]
\caption{How each of the 13 upstream functions or handlers maps to its
MorphLang program, giving the state and decision preserved, where MorphLang expresses
it, and the principal change or omission. The six derived variants are not
listed.}
\label{tab:oss-translation-map}
\fontsize{6.5pt}{7.2pt}\selectfont
\begin{tabularx}{\columnwidth}{@{}r >{\raggedright\arraybackslash}p{0.19\columnwidth} >{\raggedright\arraybackslash}X >{\raggedright\arraybackslash}X >{\raggedright\arraybackslash}X@{}}
\toprule
\# & Upstream routine & Preserved & MorphLang mapping & Changed or omitted \\
\midrule
1 & Contiki-NG RPL-lite parent selection & Select the best current
RPL parent and replace a removed parent. & \texttt{Reconfigure} scans current
neighbors and assigns \texttt{parent}, and \texttt{Send} uses it. & A core-id
ordering replaces the objective function, and the callback delay is removed. \\

2 & Contiki-NG RPL-lite local repair & Reset DAG state and select a new
parent after a newer-version packet. & \texttt{Receive} records the packet, and
\texttt{Reconfigure} resets rank and assigns \texttt{parent}. & The full DAG
and objective function are replaced by scalar state and a neighbor scan. \\

3 & Contiki-NG Trickle reset & Reset the Trickle interval
and redundancy counter under the upstream guard. & \texttt{Reconfigure}
performs the reset, and \texttt{Sample} advances the timer. & OS timer objects and
random transmission time are replaced by deterministic ticks. \\

4 & Contiki-NG RPL-classic rank update & Test and clear per-parent
feedback state before re-selection. & \texttt{Sample} tests \texttt{flagged}
and scans neighbors, and no \texttt{Reconfigure} is added. & One flag replaces
per-entry flags, and the downstream DAG-selection calls are reduced to a scan. \\

5 & TinyOS CTP route update & Candidate scan, current-parent test,
ETX threshold, and parent commit. & \texttt{Sample} contains the scan, and
\texttt{Receive} records beacon input. & nesC tasks, link-estimator storage,
and framework timers are replaced, and the periodic trigger is retained. \\

7 & OpenThread router-table update & Normalize route cost,
compare incumbent and candidate, and update or invalidate the next hop. &
\texttt{Receive} records advertised cost, and \texttt{Reconfigure} scans current
neighbors and assigns \texttt{nexthop}. & Route TLVs and router-table objects
are represented by scalar cost and one tracked destination. \\

8 & OpenThread parent-search timer & Timer and RSS conditions
that decide whether to search for a better parent. & \texttt{Sample} applies
the conditions and scans neighbors, and no \texttt{Reconfigure} is added. & The
multi-minute timer is scaled, and router entries and RSS history are reduced to
scalar inputs. \\

9 & OpenThread leader advertisement & Detect a partition-id
change and replace cached routing state. & \texttt{Receive} compares partition
ids, and \texttt{Reconfigure} selects the current destination. & Partition weighting,
attach state, RLOC16 addressing, and the address-resolver table are omitted. \\

10 & Zephyr Bluetooth Mesh friend timeout & Establish a friend from received
                                            messages and clear it after the poll
                                            timeout. & \texttt{Receive} assigns
\texttt{friendid}, \texttt{Sample} advances the timeout, and no
\texttt{Reconfigure} is added. & Credential exchange and buffered friendship messages are omitted. \\

11 & Zephyr IPv6 neighbor timeout & Advance NUD state and
remove an unreachable next hop after the retry limit. & \texttt{Send} enters
delay, \texttt{Sample} advances NUD, and \texttt{Receive} confirms a peer. &
The neighbor cache, route table, and solicitation packets are reduced to one
entry and counters. \\

12 & RIOT GNRC RPL parent removal & Remove the old
default route and choose the next preferred parent. & \texttt{Reconfigure}
scans current neighbors and assigns \texttt{parent}. & Objective-function
ordering is replaced by core-id ordering, and the FIB update is omitted. \\

13 & Meshtastic next-hop selection & Learn and retain a next
hop from acknowledgement or reply messages. & \texttt{Receive} assigns
\texttt{hop}, \texttt{Send} uses it, and no \texttt{Reconfigure} is added. & The
router database, retransmission policy, and channel logic are omitted. \\

15 & OONF OLSRv2 routing update & Set the pending flag, obey the update
trigger, and select the lowest-cost first hop. & \texttt{Receive} records a
metric change, and \texttt{Reconfigure} scans neighbors and assigns
\texttt{nexthop}. & The link-state database and kernel route installation are
omitted, and the rate-limit delay is removed in MorphLang. \\
\bottomrule
\end{tabularx}
\end{table}

\subsection{Manually Classified Reports of Stale Topology-Derived State}
\label{app:field-evidence}

Table~\ref{tab:field-evidence-full} lists the thirteen included reports with
the stale value, the change that failed to refresh it, and its group.
\emph{Analysis} marks the six reports with the stale-destination pattern
targeted by send-target checking (a manual applicability classification, not
checker runs on upstream code). \emph{Ordering} marks four reports for which the
ordered-repair contract can order a handler but the current analysis does not
apply, and \emph{outside} marks three in which the node receives no change
notification. All 103 candidates were screened, and 90 did not satisfy all
inclusion criteria, including five borderline cases with report-specific reasons
in the artifact. Each inclusion decision was checked against a quoted sentence
in the archived report. The artifact holds the fixed candidate set, the archive
of all 103 discussions with their hashes, the per-report answers, and the exact
supporting sentences.
Contiki-NG issues 890 and 933 are retained as corroboration of the issue-877
time-source-loop family, not counted as separate bugs.

\begin{table}[H]
\caption{The 13 included field reports. A \emph{send} failure is a later send that is not
delivered, and a \emph{read} failure is a read of a stale value or membership belief.}
\label{tab:field-evidence-full}
\scriptsize
\begin{tabularx}{\columnwidth}{@{}llXll@{}}
\toprule
Report & Stack & What went stale / trigger & Failure & Group \\
\midrule
\href{https://github.com/contiki-ng/contiki-ng/issues/2385}{\#2385} & Contiki-NG & route entry whose next hop re-parented away & send & analysis \\
\href{https://github.com/contiki-ng/contiki-ng/issues/877}{\#877} & Contiki-NG & TSCH time source, RPL parent reset & send & ordering \\
\href{https://github.com/contiki-ng/contiki-ng/issues/608}{\#608} & Contiki-NG & TSCH time source, nodes moved & send & ordering \\
\href{https://github.com/contiki-ng/contiki-ng/issues/2616}{\#2616} & Contiki-NG & \texttt{rpl\_dag\_t.min\_rank}, parent switch & read & ordering \\
\href{https://github.com/RIOT-OS/RIOT/issues/19423}{\#19423} & RIOT & downward route next hop, node re-parents & send & analysis \\
\href{https://github.com/openthread/openthread/issues/4240}{\#4240} & OpenThread & EID-to-RLOC entry, child changes parent & send & analysis \\
\href{https://github.com/openthread/openthread/issues/9409}{\#9409} & OpenThread & EID-to-RLOC entry, router ID released & send & analysis \\
\href{https://github.com/openthread/openthread/issues/8974}{\#8974} & OpenThread & EID-to-RLOC entry, child promoted to router & send & analysis \\
\href{https://github.com/openthread/openthread/issues/10737}{\#10737} & OpenThread & context ID, net data changes while detached & read & outside \\
\href{https://github.com/openthread/openthread/issues/9181}{\#9181} & OpenThread & parent link believed live, child evicted & read & outside \\
\href{https://github.com/openthread/openthread/issues/4472}{\#4472} & OpenThread & service RLOCs, server changes parent or role & send & analysis \\
\href{https://github.com/openthread/openthread/issues/7212}{\#7212} & OpenThread & parent recognition, parent drops sleepy child & send & outside \\
\href{https://github.com/zephyrproject-rtos/zephyr/issues/14063}{\#14063} & Zephyr & neighbor-table entries, neighbors become unreachable & send & ordering \\
\bottomrule
\end{tabularx}
\end{table}

\subsection{Validation Details of the Generated Loop}
\label{app:dispatch-detail}

This subsection gives the timing and snapshot details behind
Section~\ref{sec:dispatch-validation}.

\paragraph{Communication interval.}
A guard runs roughly a thousand times per phase, so a later guard can run a
repair omitted at an earlier site within microseconds and leave the same
phase-end destination. The \emph{communication interval} of a revision runs
from its commit to the last communication while the handled revision is older
than that revision. A commit can precede the read of the revision, so
communication in an already entered branch counts toward the interval without
violating ordered repair. The correct binaries reach 24 to 39 cycles. These are
empirical reference values, not semantic bounds, since mutants also change
instruction layout and alignment. Detection means an observation that differs
from the correct binary, not an independent proof of a contract violation.
Three mutants differ by both the interval and the state criteria of
Table~\ref{tab:mutation}.

\paragraph{Commits during repair.}
TLC explores 874 states of the split model for two roles and three revisions.
The variant that records the handled revision early and the variant with
snapshot drift each violate the corresponding invariant. The instruction-level
experiment uses one program whose repair queries neighbors on both sides of
helper calls. From compiler-emitted entry and completion addresses, a monitor
records the completed repair separately from the handled revision and compares
the neighbor-query values. Two further experiments exercise a commit in the same
cycle as a revision read, and sends dropped in the RTL because a departed
neighbor still occupies a table slot.

\subsection{The TLA+ Contract and the Checked Properties}
\label{app:tla-detail}

This subsection maps the semantic revisions and steps of
Section~\ref{sec:dispatch-validation} to TLA+ variables and actions, and
Table~\ref{tab:tla-results} reports the checked configurations.

The model state holds the commit counter $v$, the per-role delivered and handled
revisions $d_r$ and $h_r$, the role control states, a repair log $L$ of
$(r,d_r)$ pairs, two Boolean observers for duplicate repair and for an ordinary
branch before repair, and, for message-driven roles, whether a message is
pending. \texttt{TopologyChange} is the \textsc{Commit} rule and increments
$v$ only. \texttt{DeliverRevision}$(r)$ writes $v$ into $d_r$ independently of
the role's control state. \texttt{RevisionGuard}$(r)$ compares $h_r$ with $d_r$.
A role that is behind and declares \texttt{Reconfigure} records $(r,d_r)$ in $L$
and sets $h_r=d_r$ in the same step. Every ordinary action, \texttt{SendRun},
\texttt{SendSkip}, and \texttt{Recv}, includes this guard, so a delivery just
before a \texttt{Send} is still repaired first. The ordinary-branch observer is
set if an ordinary action runs at a nonzero delivered revision without the
matching repair record, and the ordering invariant requires it to stay false.

Three roles suffice for two to advance independently while a third receives a
message. The four- and five-role runs add mixed \texttt{Reconfigure}
declarations and queue states. Two to five changes cover a change before repair,
repeated changes while a repair is pending, and communication after a repair.

\begin{table}[h]
\caption{TLC results. Pass means that all reachable states were checked without a violation, and CEx means
the intended counterexample of a broken model.}
\label{tab:tla-results}
\scriptsize
\begin{tabularx}{\columnwidth}{@{}p{0.19\columnwidth}Xp{0.155\columnwidth}p{0.15\columnwidth}@{}}
\toprule
Configuration & Explored scope & Distinct states & Result \\
\midrule
Baseline & 3 roles, all repair-declaring, 1 message-driven, up to 4 changes &
  1,810,998 & Pass \\
Mixed & 3 roles, one not repair-declaring, up to 4 changes & 917,298 & Pass \\
Stress & 4 repair-declaring roles, up to 4 changes & 158,213,898 & Pass \\
Large & 5 roles, 3 repair-declaring, 2 message-driven, up to 2 changes &
  12,239,424 & Pass \\
Liveness & 3 roles, mixed repair declarations, up to 3 changes, weakly fair role loops
  and deliveries & 138,888 & Pass \\
Delivery & 3 roles, 3 rotating peers, up to 5 changes & 612,980 & Pass \\
\midrule
Five broken models & Handled revision recorded ahead of delivery / global-counter read / duplicate
repair / omitted branch check / stale target & 24 / 10 / 57 / 40 / 18 &
CEx depths 4 / 3 / 5 / 5 / 4 \\
\bottomrule
\end{tabularx}
\end{table}

Besides \texttt{TypeOK}, TLC checks that no role's handled revision passes its
delivered revision, that no delivery passes the commit counter, and that a role
records at most one repair per revision. It also checks that only a role
declaring \texttt{Reconfigure} records repairs, never at revision~0, and that
every ordinary action at a nonzero delivered revision has a matching repair
record (Theorem~\ref{thm:ordered-repair}). For the modeled repair that selects a
current neighbor, it checks that the cached target and every ordinary
\texttt{Send} target are current. Under weak fairness, it checks that a role
behind the current revision eventually handles it, that a pending repair
eventually runs, possibly coalesced into a newer one, and that once changes stop
every role becomes and stays current.

\subsection{Rocq Theorems}
\label{app:coq-send-target}

Table~\ref{tab:coq-map} maps each mechanized result of
Sections~\ref{sec:semantics-implementation} and~\ref{sec:validation} to its
Rocq name.

\begin{table}[h]
\caption{Paper results and the Rocq theorems and examples that establish them.}
\label{tab:coq-map}
\scriptsize
\begin{tabularx}{\columnwidth}{@{}p{0.27\columnwidth}Xl@{}}
\toprule
Result & Rocq name & File \\
\midrule
Theorem~\ref{thm:ordered-repair} & \nolinkurl{ordinary_runs_after_repair} & \nolinkurl{MorphLang.v} \\
Theorem~\ref{thm:coalescing} & \nolinkurl{reconfigure_at_most_once}, \nolinkurl{overlapping_commits_coalesce} & \nolinkurl{MorphLang.v} \\
Theorem~\ref{thm:continuity} & \nolinkurl{topology_change_preserves_application_state}, \nolinkurl{deliver_preserves_application_state} & \nolinkurl{MorphLang.v} \\
Handled never exceeds delivered & \nolinkurl{acknowledgement_not_ahead_of_delivery} & \nolinkurl{MorphLang.v} \\
Commits during repair & \nolinkurl{delivery_stable_during_repair}, \nolinkurl{ordinary_uses_completed_repair} & \nolinkurl{MorphLang_snapshot.v} \\
Theorem~\ref{thm:check-sound}(a), current expressions & \nolinkurl{fixed_dependencies_are_rejected}, \nolinkurl{current_expression_is_fresh} & \nolinkurl{MorphLang_qualifiers.v} \\
Theorem~\ref{thm:check-sound}, refresh analysis & \nolinkurl{An_sound}, \nolinkurl{An_sound_empty}, \nolinkurl{reconfigure_then_send_is_fresh} & \nolinkurl{MorphLang_refresh.v} \\
Examples of rules that do not record & \nolinkurl{conditional_refresh_is_not_definite}, \nolinkurl{possibly_empty_loop_refreshes_nothing}, \nolinkurl{overwrite_after_refresh_is_lost} & \nolinkurl{MorphLang_refresh.v} \\
Theorem~\ref{thm:preservation} & \nolinkurl{reconfigure_then_ordinary_send_has_no_fixed_origin}, \nolinkurl{reconfigure_then_ordinary_send_is_fresh}, \nolinkurl{ordinary_sound}, \nolinkurl{runs_sound} & \nolinkurl{MorphLang_refresh.v} \\
Theorem~\ref{thm:preservation}, rejection & \nolinkurl{fixed_dependenciesO_are_rejected} & \nolinkurl{MorphLang_qualifiers.v} \\
\bottomrule
\end{tabularx}
\end{table}

\subsection{Simulator, RTL, and FPGA Details}
\label{app:cross-layer}

This subsection describes the RTL link interface, the MorphSim comparison, and
the FPGA runs behind Section~\ref{sec:lower-layers}.

The RTL of Section~\ref{sec:setup}\anon[]{ (mm-node)} maps the link peripheral
into a 32-byte window. Its registers give the neighbor count, a neighbor index,
the identifier of the neighbor at that index, a receive-queue status bit, the
next received word and its source, a transmit register, and the revision. An
aligned word load of the revision register acquires that revision and its
neighbor table together, and the count, the identifiers, and the resolution of
transmit slots keep that table until the next such load. The host port writes a
staged table, commits it, injects received words, and drains a log of each
send's destination, slot, data, and revision, or its drop. The core is a
two-stage in-order RV32I pipeline, and the testbenches run under Icarus Verilog.

For each host-driven relay replacement, MorphSim rebuilds the neighbor snapshots
and increments the committed revision. The RQ1 schedule exchanges the only relay
in range, so \texttt{neighbor\-Count()} stays one while \texttt{neighborId(0)}
and the revision change. A packet-level reference model checks neighbor
discovery, FIFO receive order, unreachable-send drops, and revisions. With the same MorphLang
leaf coordinator compiled for each target, the RTL testbench gives a repair
window, from the commit to the first send to the new parent, of 81 cycles at the
first commit and 75 at re-selection, against 82 in MorphSim, and both have a 43-cycle loop period. Over all 43 loop positions at
which the commit can land, a count-preserving swap gives 72 to 114 cycles on the
RTL. On seven evaluation programs that declare a repair, both implementations
satisfy the same send checks with identical loop periods (69 or 70 cycles, 43
for one), and repair windows are 92 to 146 cycles in MorphSim and 91 to 140 on
the RTL. The 12 QEMU kernels (GCC 16.2.0, \texttt{qemu-riscv32} 7.2.22, and
native execution with sanitizers) test final values, not topology MMIO or cycle
counts. The artifact documents layouts, hashes, and tool versions.

\subsubsection{The FPGA Run}
\label{app:basys3}

The Basys~3 bitstream is built with an open-source flow (openXC7) or
Vivado~2024.2 from the unchanged core and link RTL, plus a UART host port and a
stamp unit. Each run resets the board and replays one host script that attaches
core~5 and then alternates the single neighbor between cores~5 and~9, 64 commits
in all. The stamp unit records each commit cycle
and the cycles of the following eight sends, with no record lost in any run.
Attach commits land at 15 of the 23 offsets possible after start-up, and one
attach per run cannot cover all of them. Figure~\ref{fig:basys3-windows} shows
the swap windows by loop position. The 81 and 75 cycles reported above equal
\(125-p\) at \(p=44\) and \(82-p\) at \(p=7\).

Each of the 19 programs is compiled into its own bitstream for the RQ1
reference schedule. The host replays the attachment and five replacements of that schedule, the
relays of the committed table beacon at the core on chip, and the board
reports the sends, deliveries, and last destination per phase. The
core starts on the host's first command, so its timers keep the offset from
each commit that RTL simulation of the same script gives. Each simulated phase
lasts the phase length used in MorphSim. The guarded and unguarded images of
Section~\ref{sec:rq3} run over 20 and 25 runs of 64 commits, as for the leaf.

\subsection{RQ3 Cost Values}
\label{app:rq3-values}

Table~\ref{tab:rq3-oss-cost} gives the values plotted in Figure~\ref{fig:rq3-cost}.

\begin{table}[h]
\caption{Cost of the eight of nine MorphLang and C pairs that meet the
phase-wise matching criterion, on the six-phase reference schedule. Cycles per transmission is
the number of cycles per issued link transmission.}
\label{tab:rq3-oss-cost}
\scriptsize
\begin{tabularx}{\columnwidth}{@{}lXrrrrrr@{}}
\toprule
\# & Program & \multicolumn{3}{c}{Cycles/transmission} &
\multicolumn{3}{c}{Code + data bytes} \\
  & & MorphLang & C & MorphLang/C & MorphLang & C & MorphLang/C \\
\midrule
\multicolumn{8}{@{}l}{\emph{Upstream C used unchanged}} \\
1 & \texttt{contiki-rpllite}        & 98.09  & 65.61  & 1.49 & 1{,}504 & 2{,}290 & 0.66 \\
2 & \texttt{contiki-rpllite-repair} & 85.24  & 62.15  & 1.37 & 1{,}388 & 2{,}548 & 0.54 \\
3 & \texttt{contiki-trickle-reset}  & 115.64 & 80.61  & 1.43 & 1{,}780 & 2{,}908 & 0.61 \\
4 & \texttt{contiki-rplclassic}     & 203.88 & 151.62 & 1.34 & 1{,}608 & 3{,}118 & 0.52 \\
\addlinespace[2pt]
\multicolumn{8}{@{}l}{\emph{Upstream C, rate limiter disabled}} \\
15 & \texttt{oonf-olsrv2}          & 116.77 & 137.31 & 0.85 & 1{,}524 & 3{,}152 & 0.48 \\
\addlinespace[2pt]
\multicolumn{8}{@{}l}{\emph{Manual C translation}} \\
5 & \texttt{tinyos-ctp-poll}        & 190.22 & 106.44 & 1.79 & 1{,}920 & 1{,}376 & 1.40 \\
7 & \texttt{openthread-routertable} & 105.09 & 64.20  & 1.64 & 1{,}592 & 1{,}412 & 1.13 \\
11 & \texttt{zephyr-nud}             & 129.36 & 62.35  & 2.07 & 1{,}768 & 1{,}352 & 1.31 \\
\bottomrule
\end{tabularx}
\end{table}